\documentclass[10pt, twocolumn]{article}
\usepackage{amsmath}
\usepackage[T1]{fontenc}
\usepackage[utf8]{inputenc}

\usepackage{fancyhdr}

\usepackage[inline]{enumitem}
\usepackage{multirow,booktabs}

\usepackage{times}

\usepackage[hang,small,labelfont=bf,up,textfont=it,up]{caption}
\usepackage{subcaption}
\usepackage{amsmath,amsfonts,amsthm,amssymb,bm,dsfont}	
\setlist{nolistsep} 

\usepackage{algorithmic}
\usepackage{siunitx}
\usepackage{bm,bbm}
\usepackage{caption,subcaption}
\usepackage{mathtools}
\mathtoolsset{showonlyrefs}
\usepackage{booktabs}
\usepackage{appendix}
\usepackage[colorlinks=true,allcolors=blue]{hyperref}

\usepackage{graphicx}
\usepackage[a4paper,top=2.75cm,bottom=2.75cm,left=1.65cm,right=1.65cm]{geometry}

\newcommand{\norm}[1]{\left\lVert#1\right\rVert}
\newcommand{\Pmat}{\mathbf{P}}
\newcommand{\Amat}{\mathbf{A}}

\newcommand{\Xmat}{\mathbf{X}}
\newcommand{\Ymat}{\mathbf{Y}}

\newcommand{\tti}{2 \to \infty}

\newcommand{\K}{K}
\newcommand{\T}{T}
\newcommand{\Op}{O_{\mathbb{P}}}

\newcommand{\Tduase}{\psi_n}
\newcommand{\Qn}{\mathbf{Q}_n}
\newcommand{\Q}{\mathbf{Q}}

\newcommand{\mat}[1]{\mathbf{#1}}

\makeatletter
\@ifundefined{algorithmicrequire}
  {}
  {}
\makeatother

\makeatletter
\@ifundefined{algorithmicensure}
  {}
  {}
\makeatother

\renewcommand{\vec}[1]{\boldsymbol{#1}}

\usepackage{numprint}
\npthousandsep{,}\npthousandthpartsep{}\npdecimalsign{.}
\newcommand\citet[1]{\citeauthor{#1}, \citeyear{#1}}

\newcommand{\titledoc}{Statistical hypothesis testing for differences between layers in dynamic multiplex networks}
\newcommand{\titleshort}{Statistical hypothesis testing for differences between layers in dynamic multiplex networks}

\usepackage{titlesec}
\titleformat{\section}[block]{\normalfont\Large\bfseries\raggedright}{\thesection}{1em}{}
\titleformat{\subsection}[block]{\normalfont\large\bfseries\raggedright}{\thesubsection}{1em}{}

\providecommand{\keywords}[1]{{\small{\textbf{\textit{Keywords ---}} #1}}}

\usepackage{fontawesome}
\usepackage{natbib}
\usepackage{setspace}

\usepackage{authblk}

\usepackage{mathtools,bigints}

\usepackage[ruled,linesnumbered,noend]{algorithm2e}
\SetKwInput{KwInput}{Input}
\let\oldnl\nl
\newcommand{\nonl}{\renewcommand{\nl}{\let\nl\oldnl}}
\makeatother

\usepackage{authblk}
\usepackage{natbib}

\usepackage{graphicx}
\graphicspath{{Fig/}}

\usepackage{amsmath,amssymb,amsthm,amsfonts}
\allowdisplaybreaks
\usepackage{url} 

\newtheorem{theorem}{Theorem}[section]
\newtheorem{definition}[theorem]{Definition}
\newtheorem{corollary}{Corollary}

\newtheorem{assumption}{Assumption}

\DeclareMathOperator*{\argmin}{arg\,min}

\author{Maximilian Baum}
\author{Francesco Sanna Passino}
\author{Axel Gandy}

\affil{Department of Mathematics, Imperial College London \\ 180 Queen’s Gate, SW7 2AZ, London (United Kingdom)}

\date{}

\title{\Huge\textbf{\titledoc}}

\allowdisplaybreaks

\usepackage{multibib}
\newcites{SM}{Supplementary references}

\begin{document}

\twocolumn[
\maketitle

\begin{center}
  {\small\bfseries Abstract}
\end{center}

\begin{center}
\begin{minipage}{0.85\textwidth} 
\small
With the emergence of dynamic multiplex networks, corresponding to graphs where multiple types of edges evolve over time, a key inferential task is to determine whether the layers associated with different edge types differ in their connectivity. In this work, we introduce a hypothesis testing framework, under a latent space network model, for assessing whether the layers share a common latent representation. The method we propose extends previous literature related to the problem of pairwise testing for random graphs and enables global testing of differences between layers in multiplex graphs. While we introduce the method as a test for differences between layers, it can easily be adapted to test for differences between time points. We construct a test statistic based on a spectral embedding of an unfolded representation of the graph adjacency matrices and demonstrate its ability to detect differences across layers in the asymptotic regime where the number of nodes in each graph tends to infinity. The finite-sample properties of the test are empirically demonstrated by assessing its performance on both simulated data and a biological dataset describing the neural activity of larval \textit{Drosophila}. 

\vspace*{0.5em}
\end{minipage}
\end{center}

\vspace*{2.5em}
]



\section{Introduction}
With the rise of increasingly complex graph-based data, the ability to answer fundamental statistical questions about these objects becomes increasingly relevant. 
In particular, when studying multiple graphs, a question that is of high practical importance is whether any of these observed graphs are structurally different from the others or whether they are independent realizations of the same random process \citep[see, for example,][]{Ginestet17, tang2017semiparametric, tang2017nonparametric, ghoshdastidar2020two, Chatterjee23, chen2024spectral, jin2024optimal}. 
This task is particularly important for \emph{dynamic multiplex networks}, graphs which are observed at multiple points in time and which present multiple connection types through different layers \citep[see, for example,][]{Kivela14}. 
For instance, in computer network monitoring, such analysis can reveal whether traffic patterns remain stable over time or exhibit structural shifts due to attacks or configuration changes \citep{Adams16}. Similarly, in neuroscience, testing for differences across layers of a multiplex brain network can shed light on how specific types of neural connections influence learning or behavior \citep[
][]{eschbach2020recurrent, jiang2021models}.
In this work, we formulate this problem within a \textit{statistical hypothesis testing} framework. 

Given a graph $\mathcal{G}$, in which connections between nodes occur across multiple layers and evolve over time, our goal is to determine whether the layers of this graph correspond to realizations from a shared latent structure or whether they differ systematically. More precisely, we consider the null hypothesis that, while the edges of $\mathcal{G}$ may vary over time, they arise from a common generative mechanism for all layers, against the alternative that at least one layer deviates in its underlying connectivity.
In order to address this question, we adopt tools from the spectral embedding literature, proposing a methodology which belongs to the class of spectral-based testing methods for networks \citep[][]{tang2017semiparametric, tang2017nonparametric, Dong20, chen2024spectral}. In particular, we build upon the \textit{Doubly Unfolded Adjacency Spectral Embedding} \citep[DUASE;][]{baum2024doubly} method to develop a statistical test based on its latent position estimates. 
This approach builds on work developed for pairwise tests between individual graphs based on standard adjacency spectral embedding \citep{tang2017semiparametric}, extending it to a general framework for \textit{joint} testing in multilayer and dynamic settings. 
Furthermore, we provide theoretical results describing the performance of the test as the number of nodes increases. 

To the best of our knowledge, the proposed testing mechanism is the first test of its kind, and no other procedures have been formally introduced in the literature to test specifically whether layers in dynamic multiplex graphs differ in their underlying connectivity patterns.

We begin this work with an overview of existing literature on statistical testing for graphs in Section \ref{sec:background}, before formally introducing the joint hypothesis testing problem for random graphs in Section~\ref{sec:Problem}. The testing methodology is detailed in Section~\ref{sec:Methods}, and Section~\ref{sec:Bootstrap} describes a bootstrap procedure that can be used to estimate the critical value for an $\alpha$-level test. The performance of these methods is assessed on both simulated data (\textit{cf.} Section~\ref{sec:Sims}) and  biological learning networks of larval \textit{Drosophila} (\textit{cf.} Section~\ref{sec:clayton}).

\subsection{Background and related literature}
\label{sec:background}

\subsubsection{Statistical testing for graphs}

The topic of hypothesis testing for graphs has been extensively explored in the literature. However, much of the existing work is centered around the topic of testing hypotheses relating to a single graph. Within this vein, \cite{green2022bootstrapping} explore the problem of bootstrapping to estimate the distribution of different subgraph structures when the nodes of a graph are exchangeable, while \cite{du2023hypothesis} and \cite{fan2022simple} work within the latent position model \citep[LPM;][]{hoff2002latent} framework and investigate the problem of testing whether the latent positions or community memberships of two nodes within the same graph are equal. \cite{Dong20} use spectral theory to construct a test aimed at identifying communities within a single graph. In this work, we are not interested in testing hypotheses for single nodes or quantities within a single graph, but rather we explore hypothesis testing for differences \textit{between} graphs. In this way, our framework is more similar to the two-sample hypothesis testing problem for random graphs in which one seeks to determine if two sets of graphs come from the same underlying distribution. Variations of this problem are an active area of research,  
recently explored in \cite{Ginestet17}, \cite{ghoshdastidar2020two}, \cite{chung2022valid}, \cite{chen2023hypothesis} and \cite{jin2024optimal}. 

Our work falls into the class of spectral embedding-based tests which have previously been developed for the two-graph case \citep[see, for example,][]{tang2017semiparametric, tang2017nonparametric}, and can also be extended to the setting where the number of nodes in the graphs diverges \citep{alyakin2024correcting}. Within this class we draw a distinction between \cite{tang2017nonparametric} and \cite{alyakin2024correcting} which are based on the maximum mean discrepancy approach of \cite{gretton2012kernel}, and \cite{tang2017semiparametric} which is most similar to our approach. In addition to sharing the underlying network model of the random dot product graph \citep[RDPG, \textit{cf.}][]{athreya2018} with our work, the structure of the testing procedure is also similar. In both our work and in \cite{tang2017semiparametric}, a test statistic is derived from spectral embedding estimates, and a bootstrap algorithm is used in order to generate a reference distribution and calculate a $p$-value. The existing literature on the two-sample hypothesis testing problem for graphs naturally results in tests that are suitable for a \textit{pairwise} testing framework for dynamic multiplex networks. In this work, we construct a joint test for differences between layers, 
which is designed to test a \textit{global} null hypothesis without the need to aggregate across multiple tests. To our knowledge, the problem of testing for global differences between layers of dynamic multiplex graphs using a single test statistic has not yet been explored.

\subsubsection{Latent position models}
\label{sec:background-lpm}

One of the foundational classes of network models is the family of so-called latent position models \citep[LPM;][]{hoff2002latent}. Under this framework, each node in a graph is equipped with a latent position $\vec x_i \in \mathbb{R}^d$ for $d \in \mathbb{N}$. 
The network adjacency matrix $\Amat$ is then treated as a random matrix where the connection between nodes $i$ and $j$ is a Bernoulli random variable with parameter obtained via a kernel function $\kappa:\mathbb{R}^d\times\mathbb{R}^d\to[0,1]$ of the latent positions of the nodes. More formally, each of the network edges are independent and satisfy $\mathbb{P}(\Amat_{i,j} = 1) = \kappa(\vec x_i, \vec x_j)$. When we take the kernel function to be the inner product, the model corresponds to the random dot product graph, whose latent positions can be consistently estimated up to an orthogonal transformation via a spectral decomposition of the adjacency matrix \citep[see, for example,][]{athreya2018}. This framework can be adapted to accommodate graphs that are both dynamic and multiplex. More generally, a number of procedures have been proposed to extend LPMs to multiplex graphs, such as  \cite{jing2021community,macdonald2022latent, lei2023bias} with extensions to the dynamic multiplex setting discussed in \cite{Oselio14, Durante17, Loyal23, baum2024doubly, Wang26}. In this work, we adopt the dynamic multiplex random dot product graph (DMPRDG) framework of \cite{baum2024doubly} due to its ability to construct a time-agnostic latent representation for each layer. In particular, in Definition~\ref{def:fixed-dmprdpg} we define a variant of this model for the case of fixed latent positions, and utilize 
it to construct a test for differences between layers, leveraging an extension of the theoretical results in \cite{baum2024doubly}.

\section{Problem setting} \label{sec:Problem}

In this section we provide a more precise definition of our problem setting and research question. We focus on a dynamic multiplex network with $K \in \mathbb{N}$ layers observed at $T \in \mathbb{N}$ time points, denoted 
$\mathcal{G} = (\mathcal{V}, \{\mathcal{E}_{k,t}\}_{k \in [K], t \in [T]})$, where $\mathcal{V}=[n]$ is a set of $n \in \mathbb{N}$ nodes shared across all layers and time points, with $[n]=\{1,\dots,n\}$, and $\mathcal{E}_{k,t} \subseteq \mathcal{V} \times \mathcal{V}$ is the set of edges of type $k$ at time $t$, such that $(i,j) \in \mathcal{E}_{k,t}$ if and only if nodes $i,j \in \mathcal{V}$ are connected by an edge of type $k$ at time $t$. We do not assume that edge sets are disjoint: for any $(k,t)$ and $(k^\prime,t^\prime)$, $\mathcal{E}_{k,t} \cap \mathcal{E}_{k^\prime,t^\prime}$ may be nonempty, allowing the same pair of nodes to be connected across multiple layers and times. 

As a framework for analyzing dynamic and multilayer graphs, we make use of 
statistical latent space models which impose a low-rank structure on the connection probabilities of the graph. The mathematical foundations that we adopt for this work come from the literature on the random dot product graph. Specifically, we adopt the dynamic multiplex random dot product graph (DMPRDPG) and the doubly unfolded adjacency spectral embedding (DUASE) estimator from \cite{baum2024doubly} for its parameters. 

\begin{definition}[Dynamic multiplex random dot product graph, DMPRDPG; \citet{baum2024doubly}]
\label{def:fixed-dmprdpg}
Consider a dynamic multiplex graph with $n$ nodes, $T$ time points and $K$ layers. For $d\in\mathbb{N}$, let $\mathcal{X}, \mathcal{Y}\subseteq\mathbb{R}^d$ such that $\vec{x}^\intercal\vec{y}\in[0,1]$ for any $\vec x\in\mathcal{X}$ and $ \vec{y}\in\mathcal{Y}$,  
and let $\Xmat^{1}, \dots, \Xmat^{K} \in \mathbb{R}^{n \times d}$ and $\Ymat^{1}, \dots, \Ymat^{T} \in \mathbb{R}^{n \times d}$ be a collection of fixed matrices such that $\Xmat^{k}_{i} \in \mathcal{X}$ and $\Ymat^{t}_{i} \in \mathcal{Y}$ for all $k \in[K]$, $t \in [T]$, $i \in [n]$, where $\Xmat^{k}_{i}$ and $\Ymat^{t}_{i}$ are the $i$-th row of $\Xmat^{k}$ and $\Ymat^{t}$ respectively. Define the $n\times n$ connection probability matrices for each time point and layer as $\Pmat^{k,t} = \Xmat^{k}\Ymat^{t\intercal}$ and construct the doubly unfolded probability matrix as
\begin{equation}
    \Pmat = \begin{bmatrix}
        \Pmat^{1,1} & \dots & \Pmat^{1,\T}\\
        \vdots & \ddots & \vdots \\
        \Pmat^{\K,1} & \dots & \Pmat^{\K,\T}
        \end{bmatrix} = \Xmat\Ymat^\intercal \in \mathbb{R}^{n\K \times n\T},
\end{equation}
where the tall matrices $\Xmat = [\Xmat^{1} \mid\dots\mid \Xmat^{\K}]\in \mathbb{R}^{n\K\times d}$ and $\Ymat = [\Ymat^{1} \mid\dots\mid\Ymat^{\T}]\in \mathbb{R}^{n\T \times d}$ are obtained by vertical stacking. 
Given a sequence of adjacency matrices $\Amat^{k,t} \in \{0,1\}^{n \times n}$, for $k\in[\K],\ t\in[\T]$, we define the doubly unfolded adjacency matrix $\Amat \in \{0,1\}^{n\K \times n\T}$ as
\begin{equation}
\Amat = \begin{bmatrix}
\Amat^{1,1} & \dots & \Amat^{1,\T}\\
\vdots & \ddots & \vdots \\
\Amat^{\K,1} & \dots & \Amat^{\K,\T}
\end{bmatrix}.
\label{eq:double_unfolding}
\end{equation}
We can then say that $\Amat \sim 
\mathrm{DMPRDPG}(\Xmat,\Ymat)$ if the matrix $\Amat^{k,t}$ has independent entries with distribution 
\begin{equation}
    \Amat^{k,t}_{i,j} \sim \mathrm{Bernoulli}(\Pmat^{k,t}_{i,j}),
\end{equation}
for all $i,j\in\{1,\dots,n\}$, $i \neq j,\ k \in [\K],\ t \in [\T]$. 
Equivalently, we write $\Amat\sim\mathrm{Bernoulli}(\Pmat)$ or $\Amat\sim\mathrm{Bernoulli}(\Xmat\Ymat^\intercal)$.
\end{definition}

\begin{definition}[Doubly unfolded adjacency spectral embedding, DUASE; \citet{baum2024doubly}] \label{definition:duase}
Consider a set of adjacency matrices $\{\Amat^{k,t}\}_{k\in [K],t\in [T]}$ arising from a dynamic multiplex graph, where $\Amat^{k,t}\in \{0,1\}^{n \times n}$ for all $k\in[K]$ and $t\in[T]$. Construct the doubly unfolded adjacency matrix $\Amat$ as 
described in Equation~\eqref{eq:double_unfolding}, and, for $d\in[n]$, consider the singular value decomposition 
\begin{equation}
    \Amat =  \mathbf{U D V}^\intercal + \mathbf{U}_\perp \mathbf{D}_\perp  \mathbf{V}_\perp^\intercal,
\end{equation} 
where $\mathbf{D}\in\mathbb{R}^{d\times d}$ is a diagonal matrix containing the $d$ largest singular values of $\Amat$, $ \mathbf{U}\in\mathbb{R}^{nK\times d}$ and $\mathbf{V} \in\mathbb{R}^{nT\times d}$ contain the corresponding left and right singular vectors respectively, and $ \mathbf{D}_\perp$, $ \mathbf{U}_\perp$ and $\mathbf{V}_\perp$ contain the remaining singular values, left singular vectors, and right singular vectors respectively.
Then, the doubly unfolded adjacency spectral embedding of $\{\Amat^{k,t}\}_{k\in [K],\ t\in [T]}$ into $\mathbb R^d$ is defined as:
\begin{align}
\hat{\Xmat} = \mathbf{U D}^{1/2}\in\mathbb R^{nK\times d}, & & 
\hat{\Ymat} = \mathbf{VD}^{1/2}\in\mathbb R^{nT\times d}.
\end{align}
\end{definition}

Through the DUASE estimation procedure in Definition~2, we obtain two sets of latent position estimates: $\hat{\Xmat}$, corresponding to layer-specific structure, and $\hat{\Ymat}$, corresponding to the second indexing dimension, here corresponding to time. The right embedding $\hat{\Ymat}$ captures variation across this second index, but we do not impose any additional smoothness, ordering, or temporal dependence assumptions on it. The DUASE left embedding $\hat{\Xmat}$ provides a representation of layer-specific latent structure that is invariant across this second dimension. This representation is therefore suitable for testing whether layers differ in their underlying connectivity patterns, in the sense that it aggregates information over the second index without assuming temporal invariance. In this work, we use the DUASE left embedding $\hat \Xmat = [\hat \Xmat^1 \mid \cdots \mid \hat \Xmat^K]$ to construct a test for differences between layers.

The theoretical framework that we consider consists of a sequence of DMPRDPGs with an increasing number of nodes. Concretely, we consider a sequence of latent position matrices $\Xmat_n \in \mathbb{R}^{n K \times d}$ and $\Ymat_n \in \mathbb{R}^{nT \times d}$ which are fixed but unknown. We observe a sequence of dynamic multiplex graphs resulting in adjacency matrices modeled as $\Amat_n \sim \mathrm{DMPRDPG}(\Xmat_n, \Ymat_n)$, and formulate the testing problem in terms of the layer-specific left latent positions $\Xmat_n$. The graphs in this sequence are made up of an increasing number of nodes $n$, and we are interested in studying the performance of our test as the number of nodes in each graph grows. In this framework, the sequence of null hypotheses $H^n_0$ of no differences between layers can be naturally defined by:
\begin{equation}
    \label{eq:DUASE_null}
    H_0^n: \Xmat_n^1 = \Xmat_n^2 = \cdots = \Xmat_n^K.
\end{equation}
Each null hypothesis $H^n_0$ is tested against an alternative hypothesis $
H^n_1$ stating that $\Xmat^k_n \neq \Xmat_n^\ell$ for at least one pair $(k,\ell)$, with $k\neq\ell,\ k,\ell\in[K]$. We note that hypotheses $H^n_0$ and $H^{n^\prime}_0$ need not be nested or related for $n \neq n^\prime$. 
In Section~\ref{sec:Methods}, we propose a test statistic for $H_0^n$ versus $H_1^n$, and we study its asymptotic behavior when the number of nodes $n$ is allowed to grow. For our results to hold, we simply require two mild regularity conditions on the sequence of latent positions, which will be further detailed in Section~\ref{sec:Methods}.

In contrast to many standard testing problems, the parameter space of the matrices $\Xmat_n$ on which we formulate our null hypothesis increases in size as $n$ increases. Our testing problem is therefore not parametric in the traditional sense. However, we treat the latent position matrices as fixed and the model places distributional requirements on the doubly unfolded adjacency matrices $\Amat_n\in\{0,1\}^{nK\times nT}$ and in this way the procedure cannot be considered strictly non-parametric. Our test is therefore best classified as semiparametric in the sense of \cite{tang2017semiparametric}. 

We note that the problem of testing for differences over time for dynamic multiplex graphs is closely related to our research question, and can be solved using the same fundamental methodology, by replacing the left embedding $\Xmat_n$ with the right DUASE embedding $\Ymat_n$. For clarity, in this work we focus exclusively on the test for differences between layers, but we note that the same procedure could equally be applied to testing for differences over time. 


Following the same key arguments in this work, we expect the theoretical setting to also be adaptable to the case where $K$ and $T$ grow with $n$, provided that the rate of growth is of order $\log(n)$ or slower. Additionally, the theory for the case of sparse graphs could be addressed by including a sparsity parameter $\rho$ to control the asymptotic connection density of the network as in \cite{baum2024doubly}, but the practical implementation of the method would remain unchanged. In this work we consider the dense regime for clarity and note that the sparse regime could be analyzed following the same key steps presented below with appropriate modifications to the asymptotic rates.

\section{Methods and results}
\label{sec:Methods}

The testing framework that we introduce for the null hypothesis defined in \eqref{eq:DUASE_null} calculates a test statistic based on spectral embedding of the doubly unfolded matrix $\Amat_n$. Let $(\hat \Xmat_n, \hat \Ymat_n) = \mathrm{DUASE}(\Amat_n)$ and define the test statistic 
\begin{equation}
    \label{eq:DUASE_Teststat}
   \Tduase = \frac{1}{K\sqrt{\log n}}\sum_{k=1}^K \left\|\hat{\Xmat}_n^k - \bar{\hat \Xmat}_n\right\|_F, 
\end{equation}
where $\bar{\hat \Xmat}_n = \frac{1}{K} \sum_{k=1}^K \hat \Xmat^k_n$ and $\|\cdot\|_F$ denotes the Frobenius norm of a matrix.  The test statistic $\Tduase$ calculates a measure of the difference between each of the layer-specific embedding estimates and their average. 

Notably, this test statistic does not include a Procrustes transformation between the different embeddings $\hat \Xmat_n^k$ and $\bar{\hat{\Xmat}}_n$ as is common in other tests based on spectral embedding \citep[see, for example,][]{tang2017semiparametric}. By performing joint embedding via DUASE rather than individual embedding procedures for each layer, the latent position estimates for different layers are directly comparable to one another. This detail reduces the computational cost of calculating the test statistic and eliminates one source of statistical error. 

As a key contribution of this work, we show that the test statistic $\Tduase$ defined in~\eqref{eq:DUASE_Teststat} can be used to construct a principled and powerful test for the null hypothesis \eqref{eq:DUASE_null} when paired with a rejection region $(C, \infty )$ for an appropriately chosen $C>0$. It must be remarked that the test statistic $\psi_n$ is suitable for a \emph{global} test, and it does not by itself identify which layers differ or whether the alternative is driven by one or multiple outlying layers. When the global null is rejected and a more fine-grained interpretation is desired, pairwise layer comparisons can be used to localize the source of the difference.

\subsection{Key results}

We now introduce two results which show that the proposed statistic $\Tduase$ is able to differentiate between the null and alternative hypotheses as the number of nodes in the graph grows. First, we show in Theorem~\ref{prop:TestLevel} that when the latent positions are equal across layers, the probability of rejection via the proposed test statistic converges to $0$ as $n$ tends to infinity. Furthermore, when the divergence between layer-specific latent positions is sufficiently large, Theorem~\ref{prop:TestPower} demonstrates that the rejection probability associated with the proposed testing procedure converges to $1$.

We derive the results in this section under the setting of known embedding dimension $d$ and additionally introduce two regularity conditions which we impose on the sequence of latent position matrices $(\Xmat_n, \Ymat_n)$.

\begin{assumption}[Singular values of $\Pmat_n$]
\label{assump:Pmat-singularvals}
Consider a sequence of 
matrices $(\Xmat_n, \Ymat_n)$, $\Xmat_n\in\mathbb{R}^{nK\times d},\ \Ymat_n\in\mathbb{R}^{nT\times d}$, and define the matrix $\Pmat_n = \Xmat_n \Ymat_n^\intercal$. We require that for all $n$ the singular values of $\Pmat_n$ are 
unique and 
$\sigma_\ell(\Pmat_n)=\Theta(n)$ 
for all $\ell\in[d]$, 
where $\sigma_\ell(\mathbf{M})$ is the $\ell$-th largest singular value of the matrix $\mathbf{M}$.
\end{assumption}

\begin{assumption}[Convergence of 
Gram matrices]
\label{assum:Q_sing_values}
Consider a sequence of 
matrices $(\Xmat_n, \Ymat_n)$, $\Xmat_n\in\mathbb{R}^{nK\times d},\ \Ymat_n\in\mathbb{R}^{nT\times d}$. We require the existence of positive definite matrices $\Delta_X,\Delta_Y \in\mathbb{R}^{d \times d}$ 
such that $n^{-1}\Xmat_n^\intercal \Xmat_n \to \Delta_X$ and $n^{-1}\Ymat_n^\intercal \Ymat_n \to \Delta_Y$.  
\end{assumption}


The prescribed growth rate of the singular values of $\Pmat_n$ in Assumption~\ref{assump:Pmat-singularvals} is a mild condition which follows naturally from the interpretation of the rows of $\Xmat_n$ and $\Ymat_n$ as latent positions with magnitudes not dependent on the total number of nodes in the graph. Assumption \ref{assum:Q_sing_values}
corresponds to a requirement that the latent positions for individual nodes exhibit a form of regularity as $n$ increases and ensures that the alignment between the given latent position sequence and the theoretical spectral embedding is stable. These assumptions are automatically satisfied in the setting where the latent positions for each node are treated as 
samples from a shared distribution \citep[\textit{cf. Propositions~3~and~12,}][]{baum2024doubly}.  


Under Assumptions~\ref{assump:Pmat-singularvals}~and~\ref{assum:Q_sing_values}, we obtain the following key theoretical results on the asymptotic behavior of the proposed test statistic $\Tduase$ when combined with a test that rejects $H_0^n$ if its observed value is larger than a constant $C$. 

\begin{theorem}
\label{prop:TestLevel}
Let $(\Xmat_n, \Ymat_n)$, $\Xmat_n\in\mathbb{R}^{nK\times d},\ \Ymat_n\in\mathbb{R}^{nT\times d}$ be a sequence of latent positions satisfying the conditions in Definition~\ref{def:fixed-dmprdpg} and Assumptions~\ref{assump:Pmat-singularvals}~and~\ref{assum:Q_sing_values}, such that $\Xmat_n^1 = \Xmat_n^2 = \cdots = \Xmat_n^K$ for all $n\in\mathbb{N}$. Let $\Amat_n \sim \mathrm{DMPRDPG} (\Xmat_n, \Ymat_n)$.  Then, there exists a constant $C > 0$ such that  
\begin{equation}
    \mathbb{P}\left( \Tduase > C \right) \overset{n\to \infty}{\longrightarrow} 0.
\end{equation}
\end{theorem}

\begin{theorem}
\label{prop:TestPower}
Let $(\Xmat_n, \Ymat_n)$, $\Xmat_n\in\mathbb{R}^{nK\times d},\ \Ymat_n\in\mathbb{R}^{nT\times d}$ be a sequence of latent positions satisfying the conditions in Definition~\ref{def:fixed-dmprdpg} and Assumptions~\ref{assump:Pmat-singularvals}~and~\ref{assum:Q_sing_values}, such that $K^{-1} \sum_{k=1}^K \norm{ \Xmat_n^k - \bar \Xmat_n}_F = \omega (\sqrt{\log n})$, where $\bar{\Xmat}_n=K^{-1}\sum_{k=1}^K \Xmat_n^k$. Let $\Amat_n \sim \mathrm{DMPRDPG} (\Xmat_n, \Ymat_n)$. Then, for any constant $C > 0$, we have 
\begin{equation}
    \mathbb{P}(\Tduase > C) \overset{n\to \infty}{\longrightarrow} 1.
\end{equation}
\end{theorem}


The proofs for Theorems \ref{prop:TestLevel} and \ref{prop:TestPower} are provided in the appendix in Sections \ref{sec:proof-level} and \ref{sec:proof-power}. Together, these results ensure the existence of a fixed rejection region such that, as $n$ tends to infinity, it is extremely unlikely that $\psi_n$ lies within this region under the null hypothesis and extremely likely to lie within this region when the difference between layer-specific embeddings is sufficiently large. We remark that the set of alternatives over which the test statistic is able to differentiate from the null is not entirely complementary to the null hypothesis. For alternatives in which the divergence of the layer-specific latent positions from the mean, $K^{-1} \sum_{k=1}^K \norm{ \Xmat_n^k - \bar \Xmat_n}_F$, scales as or slower than $\sqrt{\log n}$, it is not guaranteed that the test will reject even for large graphs. This case roughly corresponds to the setting in which fewer than $\log n$ nodes exhibit differentiated behavior across layers. The inability to detect differences affecting only a small number of nodes is expected as the test statistic $\psi_n$ captures both the true differences in layer-specific latent positions as well as the noise corresponding to each of the latent position estimates. 
When only a tiny fraction of the nodes exhibit differentiated behavior across layers, the relatively small signal of true latent position differences is drowned out by the noise from an increasing number of latent position estimates. Violations of the null hypothesis which are highly localized are naturally more difficult to detect. Alternative approaches based on different metrics, such as the two-to-infinity norm, could potentially be better suited to this problem than Frobenius norm aggregation, but would likely be less powerful in the setting where many nodes are differentiated and could require much larger graphs to be effective.  



Theorems~\ref{prop:TestLevel}~and~\ref{prop:TestPower} arise from consistency results in Theorem~1 in \cite{baum2024doubly}, adapted to the case of \textit{fixed} latent positions, resulting in the key statement in Theorem~\ref{thm:fixedbound} below. 
To place a bound on the random latent position estimates, we adopt the concept of overwhelming probability \citep{tao2010random} and 
use $\Op(\cdot)$ to denote an asymptotic rate that bounds a sequence of random variables with overwhelming probability. 
For a sequence of real-valued random variables $Z_n$ and real-valued function $f$, we write $\vert Z_n\vert = \Op\{f(n)\}$ if, for any $\gamma>0$, there exist $n_\gamma\in\mathbb{N}$ and $C_\gamma>0$ such that $\mathbb{P}\{\vert Z_n\vert \leq C_\gamma f(n)\}\geq 1 - n^{-\gamma}$ for all $n\geq n_\gamma$.

\begin{theorem}[Fixed position two-to-infinity norm bound]
\label{thm:fixedbound}

Let $(\Xmat_n, \Ymat_n)$, $\Xmat_n\in\mathbb{R}^{nK\times d},\ \Ymat_n\in\mathbb{R}^{nT\times d}$ be a sequence of latent positions satisfying the conditions in Definition~\ref{def:fixed-dmprdpg} and Assumptions~\ref{assump:Pmat-singularvals}~and~\ref{assum:Q_sing_values}. 
Let $\Amat_n \sim \mathrm{DMPRDPG} (\Xmat_n, \Ymat_n)$ and define $(\hat \Xmat_n, \hat \Ymat_n) = \mathrm{DUASE}(\Amat_n)$. Then there exists a sequence of matrices $\Qn \in \mathrm{GL}(d)$ such that for $k \in [K]$
\begin{equation}
\frac{1}{K}  \sum_{k=1}^K \left\| \hat \Xmat_n^k \Qn - \Xmat_n^k\right\|_{\tti} = \Op\left(\sqrt{\frac{\log n}{n}}\right),
\end{equation}
where $\mathrm{GL}(d)$ is the general linear group of degree $d$ and $\norm{\cdot}_{\tti}$ is the two-to-infinity matrix norm.
\end{theorem}

The proof of this result is reported in Section~\ref{sec:proof_th3}. 
Theorem \ref{thm:fixedbound} is critical for the analysis of the asymptotic behavior of the test statistic. Because the $\norm{\cdot}_{\tti}$ norm corresponds to the maximum Euclidean row norm \citep{cape19two}, Theorem~\ref{thm:fixedbound} provides a uniform bound on the maximum estimation error of the latent position estimate for any individual node. This result allows us to bound the growth rate of the noise component of $\psi_n$ as the number of nodes in the graph tends to infinity.


\subsection{Estimating the test statistic distribution via bootstrapping} \label{sec:Bootstrap}

Theorems \ref{prop:TestLevel} and \ref{prop:TestPower} demonstrate that the test statistic $\Tduase$ produces a powerful and well-controlled test when paired with an appropriate critical value $C$. Unfortunately, despite the availability of a central limit theorem for the DUASE estimators \citep[\textit{cf.}][]{baum2024doubly}, the asymptotic distribution of the test statistic $\psi_n$ under the null hypothesis is not available in closed form in practice. Therefore, in order to estimate an appropriate critical value, we propose a bootstrap procedure  which approximates the distribution of $\psi_n$ under the null hypothesis. Using this bootstrapped distribution, we are able to calculate a $p$-value for the observed test statistic. The proposed procedure is described in Algorithm~\ref{alg:DUASE_boot}.

\begin{algorithm}[h]
\caption{Bootstrap procedure for DUASE.}
\label{alg:DUASE_boot}
 \textbf{Input} Left DUASE $\hat \Xmat=[\hat \Xmat^1\mid\cdots \mid\hat \Xmat^K]\in\mathbb{R}^{nK\times d}$, right DUASE $\hat{\Ymat}=[\hat \Ymat^1\mid\cdots \mid\hat \Ymat^T]\in\mathbb{R}^{nT\times d}$, number of bootstrap samples $n_{boot}$.\\
 \textbf{Output} Test $p$-value. \\
      \textbf{Compute} $\bar{\hat{\Xmat}} = K^{-1} \sum_{k=1}^K
      \hat{\Xmat}^k$. \\
      \textbf{Compute} $\psi_{obs} = (K \sqrt{\log n})^{-1}\sum_{k=1}^K \|\hat \Xmat^k - \bar{\hat{\Xmat}}\|_F$. \\
     \For{$b \in  [n_{boot}]$}{ 
         \For{$k \in [K]$}{ 
             \For{$t \in [T]$}{
                \textbf{Sample} $\Amat^{k,t}_{boot} \sim \mathrm{Bernoulli}(  \bar{\hat{\Xmat}} {\hat{\Ymat}^{t \intercal}}) $.
             }
        } 
        \textbf{Compute} $(\hat \Xmat_{boot}, \hat \Ymat_{boot}) = \mathrm{DUASE} \left(\Amat_{boot} \right)$. \\ 
        \textbf{Compute} $\bar{\hat{\Xmat}}_{boot} = K^{-1} \sum_{k=1}^K \hat{\Xmat}_{boot}^k$. \\
        \textbf{Compute} $\psi^\ast_b = (K\sqrt{\log n})^{-1}\sum_{k=1}^K \|\hat \Xmat_{boot}^k - \bar{\hat{\Xmat}}_{boot}\|_F$.}
     \textbf{Compute} $p_{test} = (1 + n_{boot})^{-1}  [1+ \sum_{b=1}^{n_{boot}}  \mathbbm{1} \{ \psi^\ast_b > \psi_{obs}\} ]$. \\        
     \textbf{Return} Test $p$-value $p_{test}$.
\end{algorithm}

The use of bootstrap methods to estimate the distribution of graph attributes is an active research area in the random graphs literature \citep[see, for example,][]{green2022bootstrapping,Zu24,levin2025bootstrapping,dilworth2024valid}. \cite{levin2025bootstrapping} show that bootstrap methods that use spectral embedding to estimate the connection probability matrix and then simulate new graphs from this matrix are appropriate in an asymptotic sense. Finite sample behavior of such bootstraps, however, is not always well behaved. \cite{dilworth2024valid} introduce a testing procedure to determine if a bootstrap is valid and find that a nearest-neighbor based bootstrap algorithm can outperform the naive plug-in based bootstrap for graphs of finite size. For simplicity and computational tractability, in Algorithm \ref{alg:DUASE_boot} we adopt a procedure based on the standard plug-in method. Our procedure is in the spirit of the bootstrap procedure presented in \cite{tang2017semiparametric}, where we substitute a different embedding procedure and add an 
averaging step across the layer-specific embeddings. 

The logic behind the 
bootstrap procedure in Algorithm~\ref{alg:DUASE_boot} is that when the latent positions for all layers are equal, the average embedding 
will be a good estimate for each of these latent positions. When this is not the case, averaging over the layers will produce an adjacency matrix with less variation between the layers than the original graph.

\section{Simulations}
\label{sec:Sims}

To complement the asymptotic results in Section \ref{sec:Methods}, we explore an application of the testing procedure to the setting where $n$ is finite via simulation. In particular, we demonstrate that the bootstrap procedure in Algorithm \ref{alg:DUASE_boot} produces a faithful approximation of the theoretical test statistic distribution resulting in critical values that respect the defined level of the test and produce meaningful power. As an illustrative example, we consider the case of the stochastic blockmodel \citep[SBM;][]{holland1983stochastic} and adapt it to the dynamic multiplex case as the dynamic multiplex stochastic blockmodel \citep[DMPSBM;][]{baum2024doubly}, a special case of the DMPRDPG. Under this framework, each node is assigned to one layer-specific and one time-specific community, and the probability of a link between nodes $i$ and $j$ depends only on their community memberships.

\begin{definition}[Dynamic multiplex stochastic blockmodel, DMPSBM; \citet{baum2024doubly}]
    Assume that, for a dynamic multiplex network with $\K$ layers and $\T$ time points, nodes in a graph are assigned to groups or communities, where integers $z^k_i\in[G_1],\ \upsilon^t_i\in[G_2],\ G_1,G_2\in\mathbb N,\ i\in[n]$, denote the group membership assigned to the $i$-th node for the $t$-th time point and $k$-th layer respectively.
    Furthermore, define matrices $\mat{B}^{k,t} \in [0,1]^{G_1 \times G_2}$ representing between-group connection probabilities for the $k$-th layer and $t$-th time point, and set $\mathcal{B} = \{\mat{B}^{k,t}\}_{k \in [K], t \in [T]}$, $\mathcal Z=\{z_i^k\}_{i\in[n],k\in[\K]}$ and $\mathcal U=\{\upsilon_i^t\}_{i\in[n],t\in[\T]}$. For a set of adjacency matrices $\{\Amat^{k,t}\}_{k \in [K], t \in [T]}$, we say that $\Amat \sim \mathrm{DMPSBM}(\mathcal{B}, \mathcal{Z}, \mathcal{U})$, where $\Amat$ is the doubly unfolded adjacency matrix \eqref{eq:double_unfolding}, if
    \begin{equation}
        \Amat^{k,t}_{i,j} \sim \mathrm{Bernoulli}\left( \mat{B}^{k,t}_{z_i^k, \upsilon_j^t} \right)
    \end{equation}
    for each $i,j \in [n]$, $k \in [K]$ and $t \in [T]$.
\end{definition}

We simulate data from a directed DMPSBM with $G_1=2$, $G_2=2$, $K=10$ and $T=3$, where the group memberships are static such that $z^k_i = z^{k^\prime}_i$ and $\upsilon_i^t = \upsilon_i^{t^\prime}$ for all $i\in[n]$, $k,k^\prime \in [K]$ and $t, t^\prime \in [T]$. The connection probabilities between and within each community are defined as functions of $k \in [K]$ and $t \in [T]$ and are prescribed by the following set of matrices $\{\mathbf{B}^{k,t}\}$: 
\begin{equation}    
        \mathbf{B}^{k,t} =
    \begin{bmatrix}
    0.25 + \varepsilon k & 0.1 + 0.1\sin(2\pi t/T) \\ 0.1 + 0.1\sin(2\pi t/T)  & 0.25
    \end{bmatrix}.
    \label{eq:Bmat}
\end{equation}
Although the network behavior varies across both layer and time, our aim is to isolate and test for differences between layers only. In our simulation, the magnitude of the differences between layers of the graph is encoded by a parameter $\varepsilon$. The case where $\varepsilon =0$ corresponds to the null hypothesis and, as the parameter increases, we deviate further from the null setting. While varying the number of nodes in each graph, as well as the value of $\varepsilon$, we run simulations to estimate the power of the testing procedure when $\varepsilon > 0$ and demonstrate that the number of false positives is well controlled when $\varepsilon=0$. To achieve this, we generate one graph from the DMPSBM and embed it according to the DUASE procedure with embedding dimension $d=2$ corresponding to the true dimension of the unfolded matrix of connection probabilities when $\varepsilon=0$. We then input these embedding estimates into Algorithm \ref{alg:DUASE_boot} with the number of bootstrap samples $n_{boot} = 1000$ to test for differences between layers using $\alpha=0.05$. For each value of $n\in\{50,100,200,300\}$ and $\varepsilon\in\{0, 0.005, 0.01, 0.02\}$, this process is repeated $1000$ times resulting in $1000$ Monte Carlo replicates. For each value of $n$ and $\varepsilon$, the fraction of rejected tests is reported in Table \ref{tab:DUASE_power}. 
 
\begin{table}[t]
    \centering
    \sisetup{
        table-number-alignment = center,
        table-figures-integer = 1,
        table-figures-decimal = 3
    }
    \begin{tabular}{lllll}
        \toprule
        & {$\varepsilon = 0.0$} & {$\varepsilon = 0.005$} & {$\varepsilon = 0.01$} & {$\varepsilon = 0.02$} \\
        \midrule
        $n=50$  & 0.031 & 0.117 & 0.580 & 1.0 \\
        $n=100$ & 0.053 & 0.425 & 0.999 & 1.0 \\
        $n=200$ & 0.046 & 0.997 & 1.0   & 1.0 \\
        $n=300$ & 0.056 & 1.0   & 1.0   & 1.0 \\
        \bottomrule
    \end{tabular}
    \vspace*{0.5em}
    \caption{Fraction of rejected tests for different graph sizes $n$ and layer-specific differences $\varepsilon$ for the simulation in Section~\ref{sec:Sims}.}
    \label{tab:DUASE_power}
\end{table}


\begin{figure}[t]
    \centering
    \subfloat[$n=50$]{%
        \includegraphics[width=0.47\linewidth]{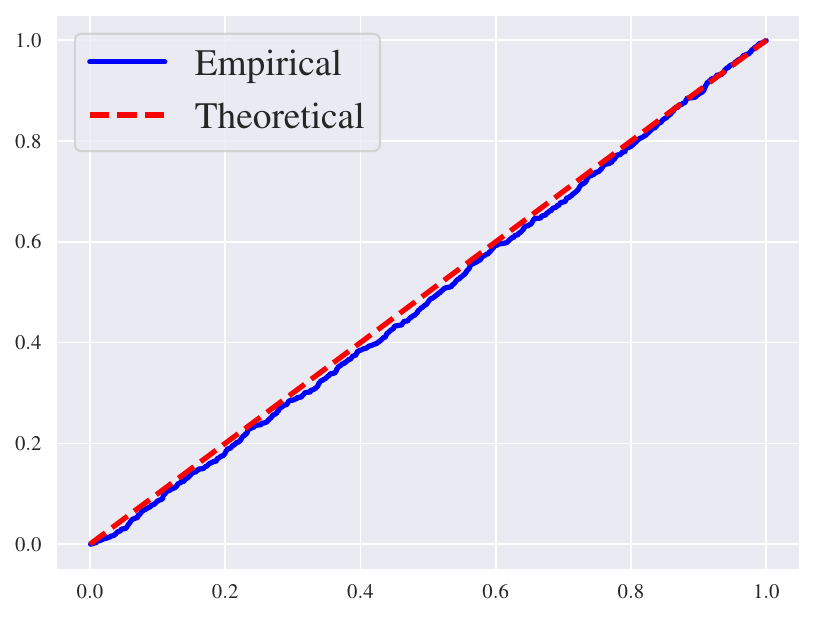}%
    }
    \hfil
    \subfloat[$n=100$]{%
        \includegraphics[width=0.47\linewidth]{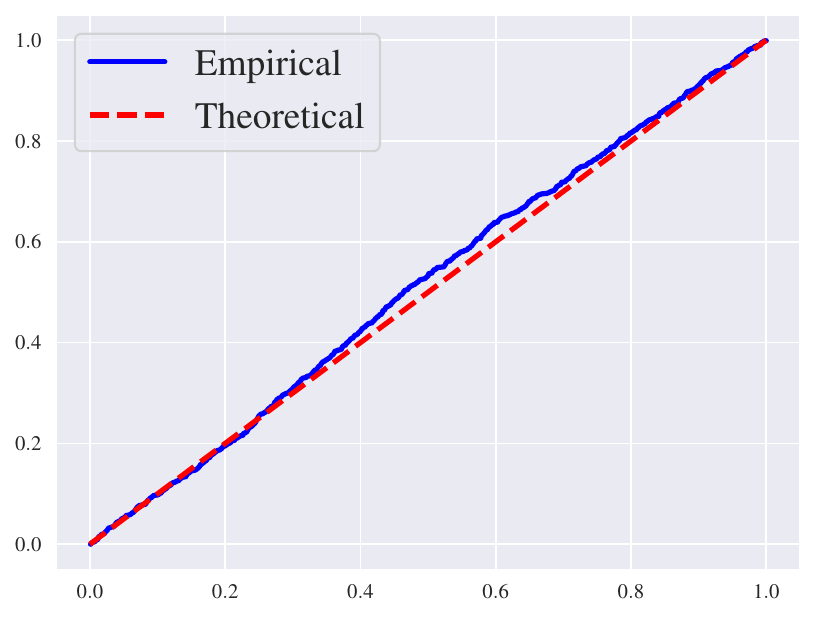}%
    }
    \vfil
    \subfloat[$n=200$]{%
        \includegraphics[width=0.47\linewidth]{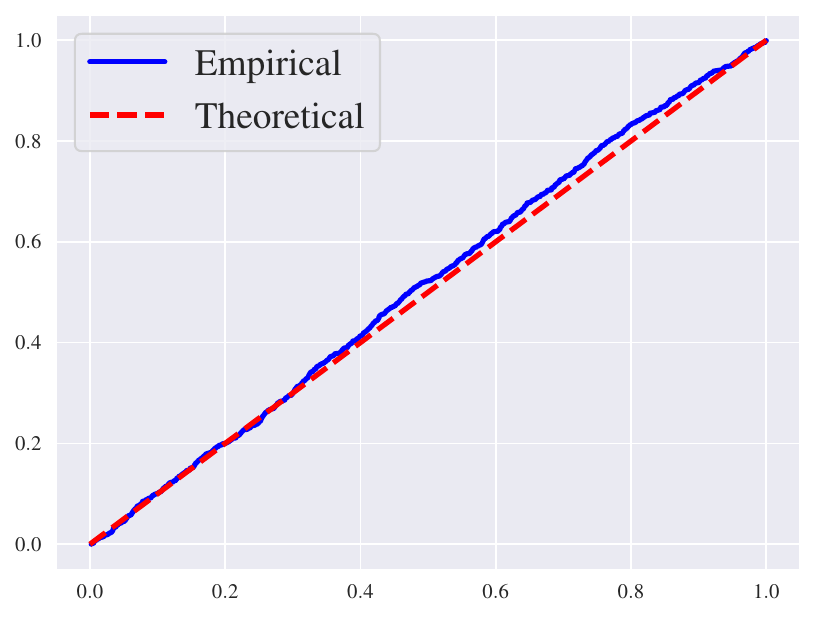}%
    }
    \hfil
    \subfloat[$n=300$]{%
        \includegraphics[width=0.47\linewidth]{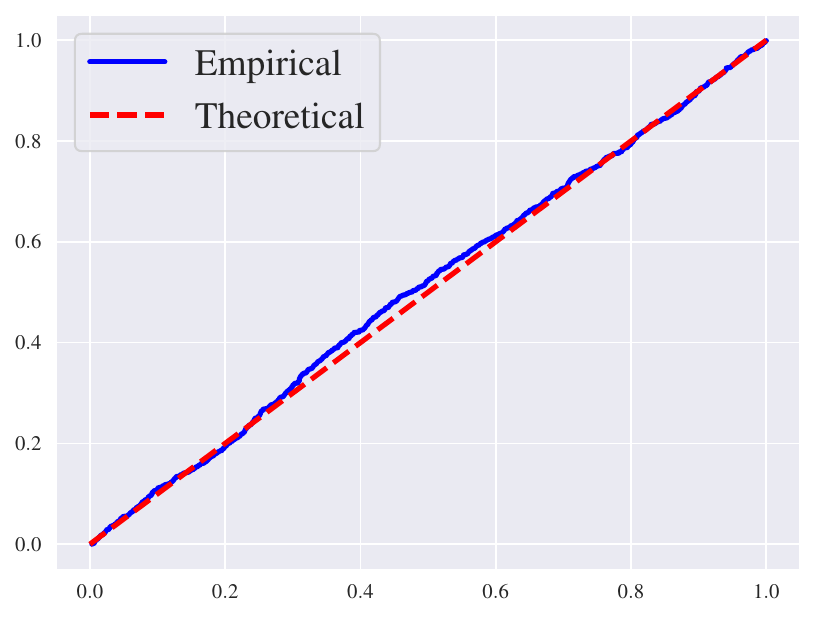}%
    }
    \caption{Cumulative distribution function of the $p$-values for graphs with $n$ nodes (for $n$ from $50$ to $300$) and $\varepsilon=0$ for the simulation in Section~\ref{sec:Sims}.}
    \label{fig:CDFs}
\end{figure}

As prescribed by Theorem \ref{prop:TestPower}, the power of the testing procedure increases as both the number of nodes in the graph and the magnitude of layer-wise differences increase. We see that under the null hypothesis, the DUASE procedure produces Type-I errors at approximately the rate of 5\% as we would expect when running the bootstrap procedure with $\alpha =0.05$. To further investigate the validity of the testing procedure under the null, we examine the distribution of $p$-values when $\varepsilon=0$ in Figure \ref{fig:CDFs}. We find that, even for graphs as small as $n=50$, the empirical distribution is close to the expected uniform distribution on $[0,1]$, suggesting that the samples generated by Algorithm \ref{alg:DUASE_boot} accurately approximate the distribution of the test statistic $\psi_n$ under the null setting. For models which are substantially more complex than the two-community DMPSBM presented here, larger graphs may be required to achieve a comparable level of calibration.

\section{Application to biological learning networks} 
\label{sec:clayton}

To demonstrate the effectiveness of the proposed methodology in an applied setting, we analyze a dataset originally presented in \cite{zheng2024dynamic}, which captures the dynamic neural activity of \textit{Drosophila} larvae over time. Using methods from \cite{eschbach2020recurrent} and \cite{jiang2021models}, it is possible to use the neural connections mapped out in \cite{winding2023connectome} to develop biologically-informed simulations of \textit{Drosophila} brain activity under different experimental conditions. 

In this analysis, we study the learning process of the \textit{Drosophila} and use our hypothesis testing method to identify which neural connections play an important role in this process. Simulations were produced to describe the neural activity of a \textit{Drosophila} in an experiment where the insect is first presented with an odor paired with a reward stimulus. At two later time points, the insect is presented with the same odor but now without the accompanying stimulus. If the insect has learned to associate the odor with the stimulus from the initial exposure, then we expect to see a neural response similar to a reward stimulus in the latter two time points despite the fact that only the odor was presented. 
To generate different experimental conditions, individual neural connections can be selectively removed from the model that is used to produce the simulated data, and this modified connectome model can then be used to simulate the brain activity during this learning process. If the removed neural connection is irrelevant for the learning process, then we expect the data from the resulting simulation to be unchanged. Only when the removed connection plays a critical role in the learning process do we expect to see a meaningful change in the simulated data. The brain activity data produced by these simulations describes the level of activation of individual neural connections and in this way is naturally represented by a graph. The experimental process of removing a single connection from the data simulation model was completed for $13$ different neural connections. For each of these, the simulation was repeated 11 times resulting in a total of $11 \times 13 = 143$ 
graphs observed at $T=160$ points in time. 

Using this collection of dynamic multiplex graphs, we implement the hypothesis testing method detailed in Section \ref{sec:Methods} to determine if any neural connections are important for the learning process and, if so, which ones. We emphasize that either the replicates within each experimental condition or the experimental conditions themselves can be treated as layers in this dataset. The testing procedure is able to accommodate both of these choices. As a first step, we perform a statistical test to confirm that the 11 replicates for each experimental condition are similar. In order to do this, we split the dataset into 13 sets of dynamic graphs and represent each of the 11 replicates as a distinct layer. We can then test for differences between these replicates by running the testing procedure, where the embedding dimension is selected using the scree-plot criterion of \cite{Zhu06}, on each of these sub-datasets. As expected, each of these tests produces a large $p$-value and for each of the 13 experimental conditions we fail to reject the null hypothesis that any of the 11 replicates are significantly different from the others. In each of these tests, the observed test statistic is smaller than the values that are generated using the bootstrap algorithm suggesting that the variation we observe between replicates is somewhat smaller than we would expect if each of these layers were truly independent samples from a DMPRDPG. This might suggest that the modeling assumptions of the DMPRDPG are only approximately satisfied for 
this application. 


 \begin{algorithm}[t]
 \caption{Modified bootstrap procedure for DUASE on averaged adjacency matrices.}
 \label{alg:Application_boot}
 \textbf{Input} Left DUASE $\hat \Xmat=[\hat \Xmat^1\mid\cdots \mid\hat \Xmat^K]\in\mathbb{R}^{nK\times d}$, right DUASE $\hat{\Ymat}=[\hat \Ymat^1\mid\cdots \mid\hat \Ymat^T]\in\mathbb{R}^{nT\times d}$, number of bootstrap samples $n_{boot}$, number of averaged adjacency matrices $n_{rep}$.\\
 \textbf{Output} Test $p$-value. \\
      \textbf{Compute} $\bar{\hat{\Xmat}} = K^{-1} \sum_{k=1}^K \hat{\Xmat}^k$. \\
      \textbf{Compute} $\psi_{obs} = (K \sqrt{\log n})^{-1} \sum_{k=1}^K \|\hat \Xmat^k - \bar{\hat{\Xmat}}\|_F$. \\
     \For{$b \in [n_{boot}]$}{ 
         \For{$k \in [K]$}{ 
             \For{$t \in [T]$}{
                 \For{$r \in [n_{rep}]$}{
                     \textbf{Sample} $\Amat^{k,t}_{boot, r} \sim \mathrm{Bernoulli}( \bar{\hat{\Xmat}} {\hat{\Ymat}^{t \intercal}})$.}
                 \textbf{Compute} $\bar \Amat^{k,t}_{boot} = n_{rep}^{-1} \sum_{r=1}^{n_{rep}} \Amat^{k,t}_{boot, r}$.}
             } 
             \textbf{Compute} $(\hat \Xmat_{boot}, \hat \Ymat_{boot}) = \mathrm{DUASE}(\bar \Amat_{boot})$. \\ 
             \textbf{Compute} $\bar{\hat{\Xmat}}_{boot} = K^{-1} \sum_{k=1}^K \hat{\Xmat}_{boot}^k$. \\
             \textbf{Compute} $\psi^\ast_b = (K \sqrt{\log n})^{-1} \sum_{k=1}^K \|\hat \Xmat_{boot}^k - \bar{\hat{\Xmat}}_{boot}\|_F$.}
     \textbf{Compute} $p_{test} = (1 + n_{boot})^{-1} [1+ \sum_{b=1}^{n_{boot}} \mathbbm{1} \{ \psi^\ast_b > \psi_{obs}\} ]$. \\        
     \textbf{Return} Test $p$-value $p_{test}$
\end{algorithm}

\begin{figure*}[t]
    \centering
    \subfloat[Global test for equality of latent positions]{%
        \includegraphics[width=0.33\textwidth]{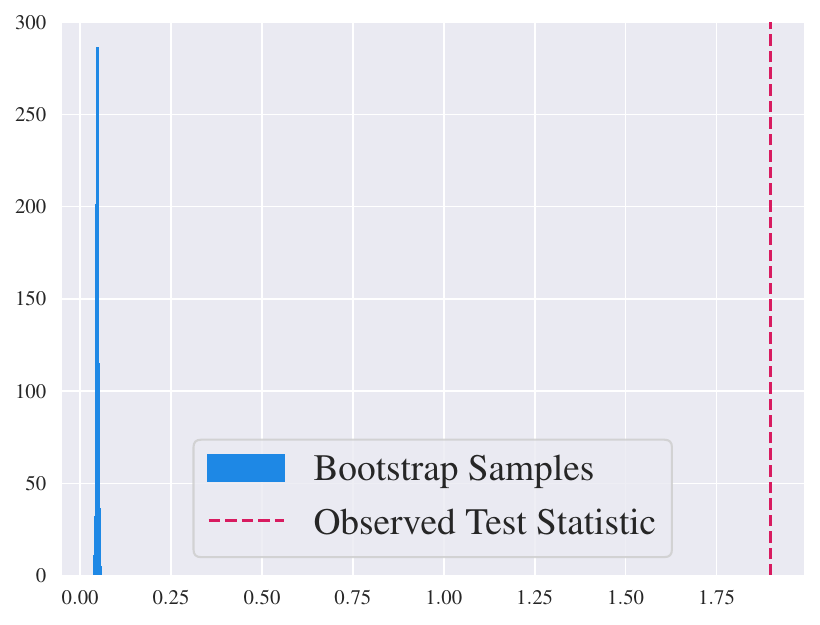}%
        \label{fig:global}
    }
    \hfil
    \subfloat[Heatmap of $p$-values for pairwise tests]{%
        \includegraphics[width=0.305\textwidth]{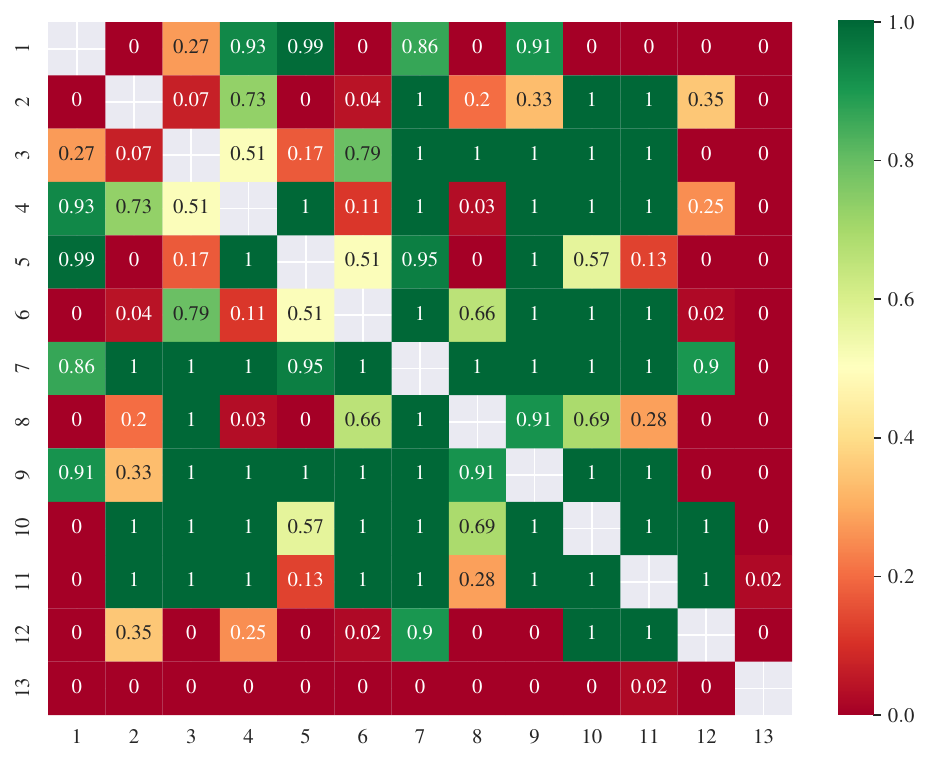}%
        \label{fig:pairwise}
    }
    \hfil
    \subfloat[Number of rejections per layer $(\alpha=0.01)$]{%
        \includegraphics[width=0.33\textwidth]{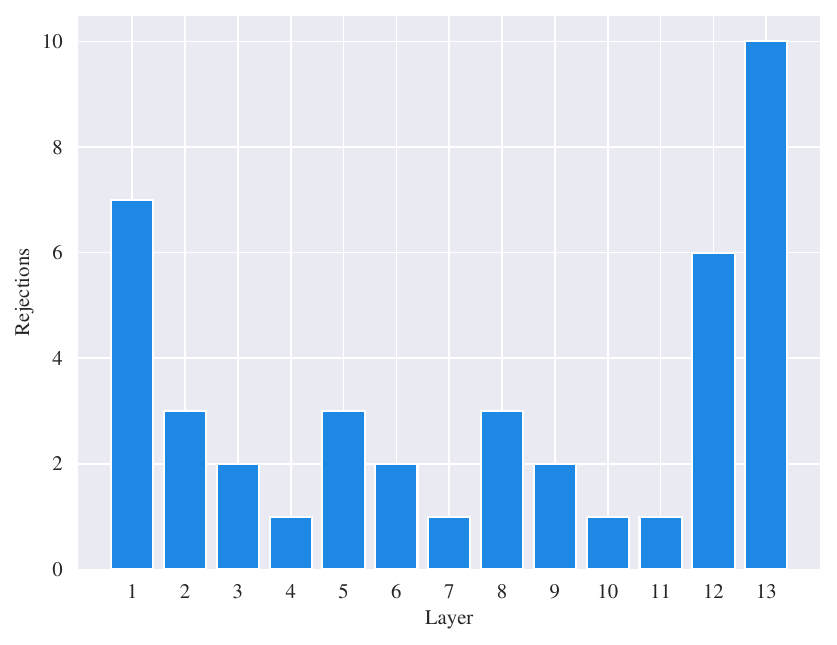}%
        \label{fig:subfigB}
    }
    \caption{Tests for differences between layers in the \textit{Drosophila} connectome data: 
    (a) global test for equality of latent positions, 
    (b) heatmap of $p$-values for all pairwise tests, and 
    (c) number of rejections per layer across pairwise tests with significance level $\alpha=0.01$.}
    \label{fig:Clayton_Pairwise}
\end{figure*}

Since the null hypothesis of no difference within replicates is not rejected, we then average the adjacency matrices corresponding to the replicates for each of the $13$ different neural connections and now treat each of the connections as a distinct layer. This results in a dataset with $K=13$ layers and $T=160$ time points describing the averaged brain activity for each of the experimental interventions. Using the modified bootstrap procedure presented in Algorithm~\ref{alg:Application_boot}, which accounts for the averaging of the adjacency matrices, we perform a global test for the null hypothesis that the latent positions for all of the $K=13$ layers are identical. This is easily rejected with the observed test statistic and bootstrap values reported in Figure \ref{fig:global}, suggesting that there are significant differences between the time series of graphs obtained by removing different neural connections. 



In order to gain further understanding of which neural connections behave most differently, we conduct each of the 78 possible pairwise tests and report the heatmap of $p$-values in Figure~\ref{fig:pairwise}, and the number of rejections at the level $\alpha=0.01$ for each layer in Figure~\ref{fig:subfigB}. Layer 13 appears to be the most differentiated and corresponds to the brain connection from neuron \textit{DAN-f1} to \textit{FBN-1}. Critically, this is the same connection that was identified as exhibiting distinct behavior in \cite{zheng2024dynamic} using clustering methods and which has a biological justification for behaving differently in this particular experiment: it is the only case in which a connection to a \textit{feedback neuron} is removed \citep{zheng2024dynamic}. 


\section{Conclusion}
\label{sec:discussion}

The ability to test for differences between layers of a dynamic multiplex graph is a research question that is of key importance and arises in fields as diverse as cybersecurity and biology. In this work, we have shown that it is possible to use spectral methods to achieve this goal. We construct a test statistic that is based on a spectral embedding of the adjacency matrices of a dynamic multiplex graph and provide two formal results to prove that this test statistic is effective in the asymptotic regime where the number of nodes tends to infinity. Concretely, we prove the existence of a rejection region such that our proposed test has a Type-I error rate that converges to $0$ under the null and power that converges to $1$ under the set of alternatives where the divergence between layer-specific latent positions scales as $\omega (\sqrt{\log n})$. In order to identify this rejection region in practice, we propose a bootstrap algorithm and demonstrate that, in simulations where a dynamic multiplex graph is generated by a two-community block model, the algorithm produces rejection regions such that the resulting test is both powerful and respects the prescribed level. Additionally, we apply this testing procedure to a biological dataset where we find that it can be effectively used to correctly detect the neural connections in the \textit{Drosophila} connectome that are most relevant for the learning process. 

A possible limitation of the present work is that the underlying DUASE does not explicitly separate layer-specific, time-specific, and layer--time interaction effects. Although the low-rank DUASE framework can implicitly capture certain interaction structures through a sufficiently large latent dimension, identifying and estimating such interactions remains beyond the scope of the current methodology.
An interesting direction for future work is also the extension of kernel and distance-based two-sample testing methods to the dynamic multiplex setting, such as the maximum mean discrepancy framework of \cite{gretton2012kernel} or energy distance-based approaches. While the test statistic proposed in this work is based on deviations of layer-specific DUASE estimates from their empirical mean rather than on kernel or distance-based machinery, it is natural to ask whether applying MMD or energy distance directly to spectral embeddings of dynamic multiplex graphs could yield complementary or improved power guarantees. Another potential avenue for future work could involve extending existing spectral-based methods \citep[such as][]{chen2024spectral} for testing for differences between layers within the dynamic multiplex case.

\section*{Code}
Code to implement the methods proposed in this work, and reproduce the simulated experiments, is available in the GitHub repository \href{https://github.com/mjbaum/dmprdpg/tree/main/testing} {\texttt{mjbaum/dmprdpg/testing}}.

\section*{Acknowledgements}

The authors thank Carey Priebe and Youngser Park (Johns Hopkins University) for providing access to the data used in Section~\ref{sec:clayton}. 
Francesco Sanna Passino acknowledges funding from the EPSRC, grant number EP/Y002113/1.

\appendix
\section{Proofs}
\label{sec:Appendix}
\subsection{Intermediate results}






\begin{corollary}
    \label{corr:DUASE-bound}
Let $(\Xmat_n, \Ymat_n)$, $\Xmat_n\in\mathbb{R}^{nK\times d},\ \Ymat_n\in\mathbb{R}^{nT\times d}$ be a sequence of latent positions satisfying the conditions in Definition~\ref{def:fixed-dmprdpg} as well as Assumptions~\ref{assump:Pmat-singularvals}~and~\ref{assum:Q_sing_values}. 
Let $\Amat_n \sim \mathrm{DMPRDPG} (\Xmat_n, \Ymat_n)$ and define $(\hat \Xmat_n, \hat \Ymat_n) = \mathrm{DUASE}(\Amat_n)$. Then there exists a sequence $\Qn \in \mathrm{GL}(d)$ such that:
        $$\frac{1}{K}  \sum_{k=1}^K \norm{\hat \Xmat_n^k \Qn - \Xmat_n^k}_F = \Op \left( \sqrt{\log n} \right).$$
    \end{corollary}
    \begin{proof}
        This result follows directly from the interpretation of $\norm{{\mathbf{M}}}_{2 \to \infty}$ as the maximum Frobenius norm of any row of $\mathbf{M}$. In particular:
        \begin{align}
            \frac{1}{K}  \sum_{k=1}^K &\norm{\hat \Xmat_n^k \Qn - \Xmat_n^k}_F = \frac{1}{K}  \sum_{k=1}^K \sqrt{\norm{\hat \Xmat_n^k \Qn - \Xmat_n^k}^2_F} \\  &\leq \frac{1}{K}  \sum_{k=1}^K \sqrt{ n \norm{\hat \Xmat_n^k \Qn - \Xmat_n^k}^2_{2 \to \infty}} \\ 
            &= \frac{1}{K}  \sum_{k=1}^K \sqrt{ n} \norm{\hat \Xmat_n^k \Qn - \Xmat^k}_{2 \to \infty}.
         \end{align}
         By applying Theorem \ref{thm:fixedbound} we have $\|\hat \Xmat_n^k \Qn - \Xmat_n^k\|_{2 \to \infty} = \Op(n^{-1/2}\log^{1/2}{n})$, which gives the result. 
    \end{proof}

\begin{corollary}
    \label{corr:ThirdTerm}
    Let $(\Xmat_n, \Ymat_n)$, $\Xmat_n\in\mathbb{R}^{nK\times d},\ \Ymat_n\in\mathbb{R}^{nT\times d}$ be a sequence of latent positions satisfying the conditions in Definition~\ref{def:fixed-dmprdpg} as well as Assumptions~\ref{assump:Pmat-singularvals}~and~\ref{assum:Q_sing_values}. 
    Let $\Amat_n \sim \mathrm{DMPRDPG} (\Xmat_n, \Ymat_n)$ and define $(\hat \Xmat_n, \hat \Ymat_n) = \mathrm{DUASE}(\Amat_n)$ and $\bar{\hat \Xmat}_n = \frac{1}{K} \sum_{k=1}^K \hat \Xmat_n^k$. Then there exists a sequence $\Qn \in \mathrm{GL}(d)$ such that:
    \begin{equation}
        \norm{\bar \Xmat_n - \bar{\hat \Xmat}_n \Qn }_F = \Op \left( \sqrt{\log{n}} \right).
    \end{equation}
\end{corollary}
\begin{proof}
The result follows from the following inequality:
    \begin{align}
    \norm{\bar \Xmat_n - \bar{\hat \Xmat}_n \Qn }_F &= \norm{\frac{1}{K} \sum_{k=1}^K \left( \Xmat_n^k - \hat \Xmat_n^k \Qn \right)}_F \\ 
    &\leq \frac{1}{K}  \sum_{k=1}^K \norm{\Xmat_n^k - \hat \Xmat_n^k \Qn}_F .
    \end{align}
    Applying Corollary \ref{corr:DUASE-bound} yields the desired result.
\end{proof}

\subsection{Proof of Theorem~\ref{thm:fixedbound}}
\label{sec:proof_th3}

\begin{proof}
The proof for this result follows from the proof for the case of random latent positions presented in Theorem~1 of \cite{baum2024doubly}, but where we replace Proposition~3 of \cite{baum2024doubly} with Assumption~\ref{assump:Pmat-singularvals}. 
\end{proof}

\subsection{Proof of Theorem \ref{prop:TestLevel}}
\label{sec:proof-level}
\begin{proof}
Let $\Qn$ denote the invertible transformation that most closely maps $\hat \Xmat_n$ to $\Xmat_n$ in the two-to-infinity norm, written $\Qn =\argmin_{\Q \in \mathrm{GL}(d)} \|\Xmat_n - \hat \Xmat_n \Q\|_{\tti}$. We begin by decomposing and bounding the numerator of the proposed test statistic $\Tduase$, rescaled by $\sigma_d(\Qn)\equiv\sigma_{min}(\Qn)$, the smallest singular value of $\Qn$:
\begin{align}
\label{eq:Decomposition}
& \frac{\sigma_{min}(\Qn)}{K}\sum_{k=1}^K \norm{\hat \Xmat_n^k - \bar{\hat \Xmat}_n}_F \leq \frac{1}{K} \sum_{k=1}^K \norm{\hat \Xmat_n^k \Qn - \bar{\hat \Xmat}_n \Qn}_F \notag \\ 
& \leq \frac{1}{K}  \sum_{k=1}^K \norm{\hat \Xmat_n^k \Qn - \Xmat_n^k}_F + \frac{1}{K} \sum_{k=1}^K \norm{ \Xmat_n^k - \bar \Xmat_n}_F \\ &\hspace*{3em} + \norm{\bar \Xmat_n - \bar{\hat \Xmat}_n \Qn }_F.
\end{align}    
Under the null hypothesis $H^n_0$ in Equation~\eqref{eq:DUASE_null}, the second term of the right-hand side is equal to zero by definition and therefore can be excluded from our analysis. We apply Corollary~\ref{corr:DUASE-bound} to bound the first term and can bound the third term by Corollary~\ref{corr:ThirdTerm}. The quantity $\sigma_{min}(\Qn)$ is of constant order by Proposition 12 of \cite{baum2024doubly} and therefore 
\begin{equation}
    \frac{1}{K} \sum_{k=1}^K \norm{\hat \Xmat_n^k - \bar{\hat \Xmat}_n}_F = \Op \left( \sqrt{\log{n}} \right).
\end{equation}
It follows that $\Tduase = \Op \left( 1 \right)$ and therefore by the definition of overwhelming probability for $n$ large and any $\alpha>0$, there exists a constant $C$ such that $\mathbb{P}\left( \Tduase > C \right) \leq n^{-\alpha}$, which gives the result.
\end{proof}

\subsection{Proof of Theorem~\ref{prop:TestPower}}
\label{sec:proof-power}

\begin{proof}
We begin with the following decomposition: 
\begin{align}  
& \norm{ \Xmat_n^k - \bar \Xmat_n}_F  \\ 
&= \norm{ \hat \Xmat_n^k \Qn - \bar{\hat \Xmat}_n \Qn + \bar{\hat \Xmat}_n \Qn - \bar \Xmat_n + \Xmat_n^k - \hat{\Xmat}^k_n \Qn }_F \\ 
&\leq \norm{ \hat \Xmat_n^k - \bar{\hat \Xmat}_n}_F \norm{\Qn}_F + \norm{ \bar{\hat \Xmat}_n \Qn - \bar \Xmat_n }_F \\ &\hspace*{3em} + \norm{ \Xmat_n^k - \hat{\Xmat}_n^k \Qn}_F .
\end{align}
Averaging over all layers, we can apply Corollary~\ref{corr:DUASE-bound} and Corollary~\ref{corr:ThirdTerm} to show that the latter two terms are $\Op\left( \sqrt{\log{n}} \right)$ and we make use of the fact that $\norm{\Qn}_F$ is of constant order by Proposition 12 of \cite{baum2024doubly}. By assumption, we have $K^{-1}\sum_{k=1}^K \norm{ \Xmat_n^k - \bar \Xmat_n}_F = \omega(\sqrt{\log n})$. This gives: 
\begin{equation}
        \frac{1}{K} \sum_{k=1}^K \norm{ \hat \Xmat_n^k - \bar{\hat \Xmat}_n}_F = \omega_{\mathbb{P}} \left(\sqrt{\log n} \right). 
    \end{equation}
It follows that $\Tduase = \omega_{\mathbb{P}} \left( 1 \right)$, and therefore for any constant $C > 0$ we have that $\mathbb{P}\left(\Tduase  > C \right)  \to 1$ as $n \to \infty$.
\end{proof}

\section{Additional simulation results}
\subsection{Effects of model misspecification}
\label{sec:Robustness}

The testing procedure defined in Section~\ref{sec:Methods} and corresponding theoretical guarantees defined in Theorems~\ref{prop:TestLevel} and~\ref{prop:TestPower} are derived under the setting where the embedding dimension $d$ is known and the edges are assumed to be independent. 
Therefore, under the simulation setting presented in Section~\ref{sec:Sims}, we perform additional simulations with $n=50$ to illustrate the ways in which the testing procedure can be impacted in the scenario where these assumptions are not satisfied. 

First, we begin by investigating the impact of a misspecified embedding dimension $d$ on the power and level of the test. Results are reported in Figure~\ref{fig:subfig-dim}. Under the misspecified regime,  
selecting a dimension that is too small reduces the power of the test, 
since the curve corresponding to $d=1$ is overly conservative. On the other hand, the remaining curves 
suggest that when the dimension is chosen to be too large, additional noise is introduced into the calculation of the test statistic, resulting in 
an inflated Type-I error 
with the degree of inflation increasing with the degree of misspecification. 

Next, to assess the impact of non-independent entries, 
we conduct simulations in which the entries of the adjacency matrix are defined by 
correlated Bernoulli random variables.
To introduce this correlation, we 
construct a factor model with $m=10$ latent factors, where each possible edge 
$\ell\in[n]\times[n]$
is assigned a randomly generated vector of factor weights $\vec{f}_\ell\in \mathbb{R}^m$, normalized such that $\norm{\vec{f}_\ell}=1$. The formation of an edge is then determined by a random variable $Z_\ell = \sqrt{a}\,\vec{f}_\ell^\intercal\boldsymbol{\eta} + \sqrt{1-a}\, \varepsilon_i$ where $\boldsymbol{\eta} \sim \mathcal{N}_m(0, \mat I_m)$ are latent factors, shared between all edges, $a \in [0,1]$ controls the proportion of variance that is shared between edges, and $\varepsilon_i \sim \mathcal{N}(0,1)$ is the independent variance component. Under this framework, each $Z_\ell$ is marginally a standard normal random variable, whereas jointly the collection of all $n(n-1)$ random variables $Z_\ell$ has a multivariate normal distribution with covariance 
$\mat{\Sigma} = a \mat{FF}^\intercal + (1-a)\mat{I}_{n(n-1)}$, where $\mat{F} \in \mathbb{R}^{n(n-1) \times m}$ is obtained by stacking the factor loadings of all edges into a matrix. 
To model the formation of the $\ell$-th edge, which occurs with probability $p_\ell$ obtained via the DMPRDPG, we then threshold the draw from the random variable $Z_\ell$ at the level of $\Phi^{-1}(p_\ell)$, where $\Phi(\cdot)$ denotes the CDF of a standard normal random variable. 
Apart from the correlation structure induced by this model, all other aspects of the simulation remain unchanged relative to 
Section~\ref{sec:Sims}. 
The results are plotted in Figure~\ref{fig:subfig-corr}, showing that increasing dependence between edges results in inflated Type-I error rates with the impact becoming prominent for larger values of $a$, which measure the strength of correlation within edges. 




\begin{figure}[t]
    \centering

    \subfloat[Misspecified dimension]{%
        \includegraphics[width=0.47\linewidth]{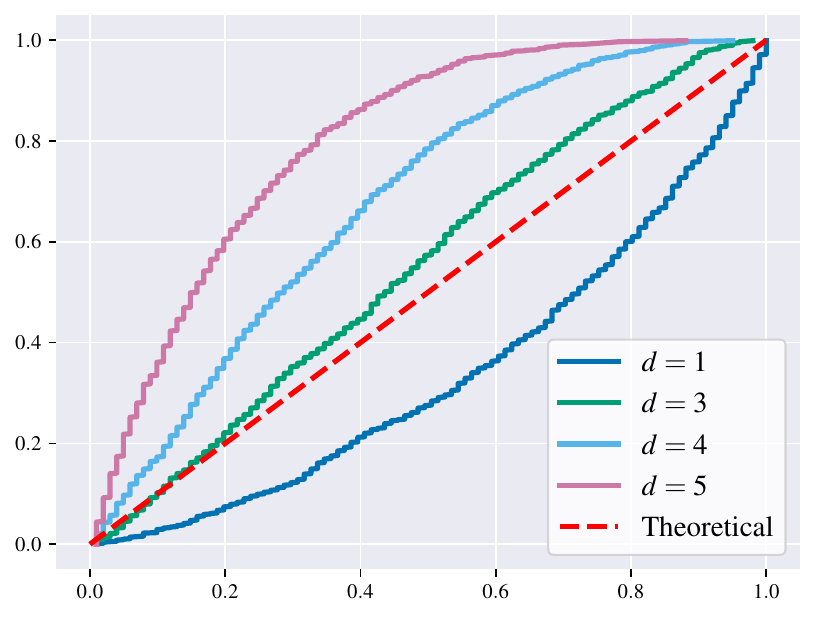}%
        \label{fig:subfig-dim}
    }
    \hfil
    \subfloat[Correlated entries]{%
        \includegraphics[width=0.47\linewidth]{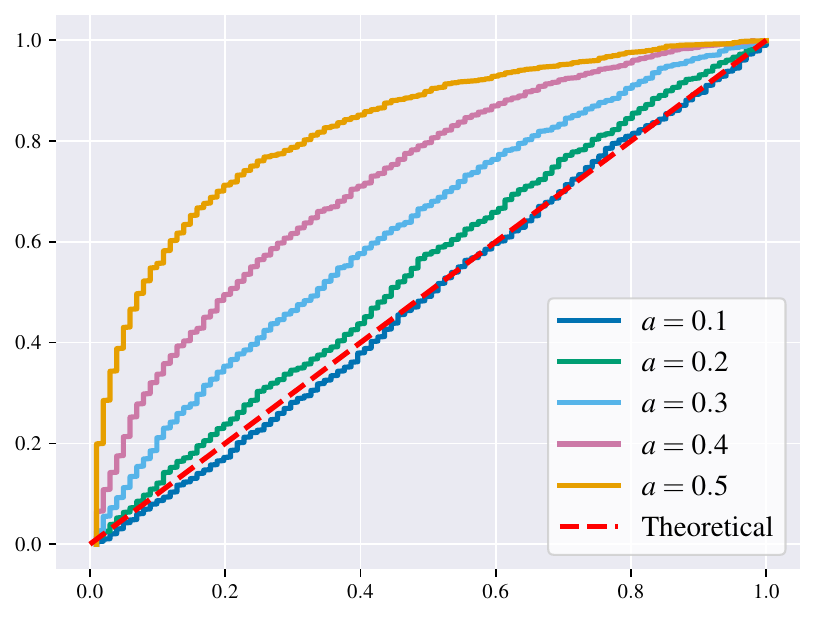}%
        \label{fig:subfig-corr}
    }

    \caption{Distribution of $p$-values under misspecified embedding dimension and correlated entries for the simulation presented in Section~\ref{sec:Sims}.}
    \label{fig:misspeccified-sims}
\end{figure}
\subsection{Comparison with existing methods}
\label{sec:Comparisons}
\subsubsection{Comparison with multiple pairwise tests}
\label{sec:mult-testing}
Although the problem of testing for differences between layers in dynamic multiplex graphs has not been explicitly studied in the literature, it is possible to 
construct alternative tests through collections of pairwise tests based on existing methodologies for the two-graph case. In order to assess the power of such a procedure, we conduct a simulation similar to that in Section~\ref{sec:Sims} by generating data from an undirected DMPSBM with $K=10$ layers and $T=3$ time points where the connection probabilities are given by
\begin{equation}    
        \mathbf{B}^{k,t} =
    \begin{bmatrix}
    0.25 + \varepsilon k & 0.6 + 0.1\sin(2\pi t/T) \\ 0.6 + 0.1\sin(2\pi t/T)  & 0.25
    \end{bmatrix}.
    \label{eq:Bmat-pair}
\end{equation}
For each $t \in [T]$ we utilize the testing procedure of \cite{tang2017semiparametric}, selecting $d$ using the method of \cite{Zhu06}, to test each of the possible pairwise hypotheses $\Pmat^{k,t} = \Pmat^{k^\prime, t}$ for $k, k^\prime \in [K]$. The $p$-values of each of the pairwise tests are aggregated via the Cauchy Combination Test \citep[CCT;][]{liu2020cauchy} to obtain a global $p$-value for the null hypothesis of equality of latent positions in all layers. 

Additionally, we compare the performance of the DUASE-based test against a procedure based on a spectral projection constructed as follows: given order-$d$ truncated SVDs $\Amat^{k,t}\approx \hat{ \mathbf{U}}^{k,t} \hat{\mathbf{D}}^{k,t} \hat{\mathbf{V}}^{k,t\intercal}$ for each $k\in[K],\, t\in[T]$, we write $\boldsymbol{\Pi}^k = T^{-1} \sum_{t=1}^{T} \hat{\mathbf{U}}^{k,t} \hat{\mathbf{D}}^{k,t} \hat{\mathbf{U}}^{k,t\intercal}$ be the average scaled projection matrices across time points for a fixed layer. We then consider the test statistic $\sum_{k=1}^{K} \|\boldsymbol{\Pi}^k - \bar{\boldsymbol{\Pi}}\|_{F}$ where $\bar{\boldsymbol{\Pi}} = K^{-1} \sum_{k=1}^{K} \boldsymbol{\Pi}^k$, with distribution estimated via bootstrapping. We call this procedure spectral projection test (SPT).  

\begin{table}[t]
    \centering
    Method: DUASE (ours) \\[0.25em]
    \sisetup{
        table-number-alignment = center,
        table-figures-integer = 1,
        table-figures-decimal = 3
    }
    \begin{tabular}{lllll}
        \toprule
        & {$\varepsilon = 0.0$} & {$\varepsilon = 0.01$} & {$\varepsilon = 0.02$} & {$\varepsilon = 0.04$} \\
        \midrule
        $n=50$  & 0.072 & 0.450 & 0.982 & 1.0 \\
        $n=100$ & 0.050 & 0.962 & 1.0 & 1.0 \\
        \bottomrule\vspace*{0.1em}
    \end{tabular}
    
    \centering
    Method: Multiple testing based on \cite{tang2017semiparametric} \\[0.25em]
    \sisetup{
        table-number-alignment = center,
        table-figures-integer = 1,
        table-figures-decimal = 3
    }
    \centering
    \begin{tabular}{lllll}
        \toprule
        & {$\varepsilon = 0.0$} & {$\varepsilon = 0.04$} &
        {$\varepsilon = 0.06$} &{$\varepsilon = 0.08$} \\
        \midrule
        $n=50$  & 0.0 & 0.013 & 0.475 & 0.890 \\
        $n=100$ & 0.0 & 0.455 & 0.986 & 1.0 \\
        \bottomrule\vspace*{0.1em}
    \end{tabular}

    \centering
    Method: Spectral projection test (SPT) \\[0.25em]
    \sisetup{
        table-number-alignment = center,
        table-figures-integer = 1,
        table-figures-decimal = 3
    }
    \begin{tabular}{lllll}
        \toprule
        & {$\varepsilon = 0.0$} & {$\varepsilon = 0.01$} & {$\varepsilon = 0.02$} & {$\varepsilon = 0.04$} \\
        \midrule
        $n=50$  & 0.061 & 0.122 & 0.308 & 0.962 \\
        $n=100$ & 0.079 & 0.292 & 0.950 & 1.0 \\
        \bottomrule\vspace*{0.1em}
    \end{tabular}
    \caption{Fraction of rejected tests for different graph sizes $n$ and layer-specific differences $\varepsilon$ for the simulation in  Section~\ref{sec:mult-testing}.}
    \label{tab:Tang_power_comp}
\end{table}

The fraction of rejected tests for the two alternative procedures is compared with that of our testing procedure in Table~\ref{tab:Tang_power_comp}. We observe that the procedure of \cite{tang2017semiparametric} can be adapted to detect differences between layers in dynamic multiplex graphs, however it is naturally not as powerful as our testing procedure proposed in Section~\ref{sec:Methods}, 
specifically designed for this problem. 
Similarly, the SPT remained less powerful than the proposed DUASE procedure. 
The conservative nature of the pairwise approach may arise from a number of sources, including the complexity of aggregating a large number of individual $p$-values as well as the ability of DUASE to borrow strength across multiple layers and time points.

\subsubsection{Comparisons with multiple replicates}
\label{sec:Replicates_comp}

Another testing method that could be adapted to test for differences between layers in dynamic multiplex graphs is the spectral-based approach of \cite{chen2024spectral}. This procedure allows for a global test of equality for connection probability matrices when multiple replicates from each of the connection probability matrices are observed. It is therefore most directly comparable with the modified bootstrap procedure detailed in Algorithm~\ref{alg:Application_boot}. To compare the effectiveness of these methods, we conduct a simulation in which $n_{rep}=50$ copies of an undirected DMPSBM are observed with underlying connection probabilities given by 
\begin{equation}    
        \mathbf{B}^{k,t} =
    \begin{bmatrix}
    0.25 + \varepsilon k & 0.1 + 0.1\sin(\pi t/T) \\ 0.1 + 0.1\sin(\pi t/T)  & 0.25
    \end{bmatrix}.
    \label{eq:Bmat-mult}
\end{equation}
At each time point $t \in [T]$, the procedure of \cite{chen2024spectral} is used to test the hypothesis that $\Pmat^{1,t} = \Pmat^{2,t} =\dots = \Pmat^{K,t}$. The $p$-values of each of these $T$ tests were then combined via the CCT. In this setting, we find that the procedure of \cite{chen2024spectral} is effective, but naturally less powerful given that it is designed for a more general setting which relies on an independent edge assumption but not a low-rank assumption on the connection probability matrices. The fraction of rejected tests for both methods are reported in Table~\ref{tab:Chen_power_reps}. In comparison to the results presented in Section~\ref{sec:Sims}, we also see that the procedure where multiple replicates are averaged is much more powerful. 
This suggests that the averaging procedure used in Section~\ref{sec:clayton} makes effective use of the additional information provided by multiple replicates. 

\begin{table}[t]
    \centering
    Method: DUASE (ours) \\[0.25em]
    \sisetup{
        table-number-alignment = center,
        table-figures-integer = 1,
        table-figures-decimal = 3
    }
    \begin{tabular}{lllll}
        \toprule
        & {$\varepsilon = 0.0$} & {$\varepsilon =0.001$} & {$\varepsilon =0.002$} & {$\varepsilon =0.004$} \\
        \midrule
        $n=50$  & 0.036 & 0.183 & 0.831 & 1.0 \\
        $n=100$ & 0.046 & 0.674 & 1.0 & 1.0 \\
        \bottomrule\vspace*{0.1em}
    \end{tabular}
    
    \centering
    Method: Multiple testing based on \cite{chen2024spectral} \\[0.25em]
    \sisetup{
        table-number-alignment = center,
        table-figures-integer = 1,
        table-figures-decimal = 3
    }
    \begin{tabular}{llllll}
        \toprule
        & {$\varepsilon = 0.0$} & {$\varepsilon = 0.002$} & {$\varepsilon = 0.004$} &
        {$\varepsilon = 0.008$} \\
        \midrule
        $n=50$  & 0.011 & 0.013 & 0.097 & 0.924 \\
        $n=100$ & 0.098 & 0.221 & 0.904 & 1.0 \\
        \bottomrule\vspace*{0.1em}
    \end{tabular}
    \caption{Fraction of rejected tests for different graph sizes $n$ and layer-specific differences $\varepsilon$ for the simulation in  Section~\ref{sec:Replicates_comp}.}
    \label{tab:Chen_power_reps}
\end{table}

\bibliographystyle{rss}
\singlespacing
\bibliography{reference}

@article{Dong20,
	author = {Zhishan Dong and Shuangshuang Wang and Qun Liu},
	journal = {Information Sciences},
	pages = {1360-1371},
	title = {Spectral based hypothesis testing for community detection in complex networks},
	volume = {512},
	year = {2020}}

@article{Ginestet17,
	author = {Cedric E. Ginestet and Jun Li and Prakash Balachandran and Steven Rosenberg and Eric D. Kolaczyk},
	journal = {The Annals of Applied Statistics},
	number = {2},
	pages = {725 -- 750},
	title = {{Hypothesis testing for network data in functional neuroimaging}},
	volume = {11},
	year = {2017}}

@InProceedings{Chatterjee23,
  title = 	 {Two-Sample Tests for Inhomogeneous Random Graphs in $L_r$ Norm: Optimality and Asymptotics},
  author =       {Chatterjee, Sayak and Saha, Dibyendu and Dan, Soham and Bhattacharya, Bhaswar B.},
  booktitle = 	 {Proceedings of The 26th International Conference on Artificial Intelligence and Statistics},
  pages = 	 {6903--6911},
  year = 	 {2023},
  editor = 	 {Ruiz, Francisco and Dy, Jennifer and van de Meent, Jan-Willem},
  volume = 	 {206},
  series = 	 {Proceedings of Machine Learning Research},
  publisher =    {PMLR}
}

@article{Kivela14,
	author = {Kivel{\"a}, Mikko and Arenas, Alex and Barthelemy, Marc and Gleeson, James P. and Moreno, Yamir and Porter, Mason A.},
	journal = {Journal of Complex Networks},
	month = {09},
	number = {3},
	pages = {203-271},
	title = {Multilayer networks},
	volume = {2},
	year = {2014}}

@book{Adams16,
author = {Adams, Niall and Heard, Nick},
title = {Dynamic Networks and Cyber-Security},
publisher = {World Scientific (Europe)},
year = {2016}
}

@article{Zu24,
	author = {Zu, Tianhai and Qin, Yichen},
	journal = {Biometrika},
	month = {09},
	number = {1},
	pages = {asae046},
	title = {Local bootstrap for network data},
	volume = {112},
	year = {2024}}

@article{tang2017semiparametric,
  title={A semiparametric two-sample hypothesis testing problem for random graphs},
  author={Tang, Minh and Athreya, Avanti and Sussman, Daniel L. and Lyzinski, Vince and Park, Youngser and Priebe, Carey E.},
  journal={Journal of Computational and Graphical Statistics},
  volume={26},
  number={2},
  pages={344--354},
  year={2017},
  publisher={Taylor \& Francis}
}

@article{athreya2018,
  author  = {Avanti Athreya and Donniell E. Fishkind and Minh Tang and Carey E. Priebe and Youngser Park and Joshua T. Vogelstein and Keith Levin and Vince Lyzinski and Yichen Qin and Daniel L Sussman},
  title   = {Statistical Inference on Random Dot Product Graphs: a Survey},
  journal = {Journal of Machine Learning Research},
  year    = {2018},
  volume  = {18},
  number  = {226},
  pages   = {1--92}
}

@article{tao2010random,
  title={Random matrices: Universality of local eigenvalue statistics up to the edge},
  author={Tao, Terence and Vu, Van},
  journal={Communications in Mathematical Physics},
  volume={298},
  pages={549--572},
  year={2010},
  publisher={Springer}
}

@article{tang2017nonparametric,
  title={A nonparametric two-sample hypothesis testing problem for random graphs},
  author={Tang, Minh and Athreya, Avanti and Sussman, Daniel L and Lyzinski, Vince and Priebe, Carey E},
  year={2017},
  journal={Bernoulli},
  volume={23},
  number={3},
  pages={1599-1630}
}

@article{baum2024doubly,
  title={Doubly unfolded adjacency spectral embedding of dynamic multiplex graphs},
  author={Baum, Maximilian and {Sanna Passino}, Francesco and Gandy, Axel},
  journal={arXiv preprint arXiv:2410.09810},
  year={2024}
}

@article{zheng2024dynamic,
  title={Dynamic networks clustering via mirror distance},
  author={Zheng, Runbing and Athreya, Avanti and Zlatic, Marta and Clayton, Michael and Priebe, Carey E},
  journal={arXiv preprint arXiv:2412.19012},
  year={2024}
}

@article{chen2023hypothesis,
  title={Hypothesis testing for populations of networks},
  author={Chen, Li and Zhou, Jie and Lin, Lizhen},
  journal={Communications in Statistics-Theory and Methods},
  volume={52},
  number={11},
  pages={3661--3684},
  year={2023},
  publisher={Taylor \& Francis}
}

@article{jin2024optimal,
  title={Optimal network pairwise comparison},
  author={Jin, Jiashun and Ke, Zheng Tracy and Luo, Shengming and Ma, Yucong},
  journal={Journal of the American Statistical Association},
  pages={1048–1062},
  volume={120},
  issue={550},
  year={2025},
  publisher={Taylor \& Francis}
}

@article{green2022bootstrapping,
  title={Bootstrapping exchangeable random graphs},
  author={Green, Alden and Shalizi, Cosma Rohilla},
  journal={Electronic Journal of Statistics},
  volume={16},
  number={1},
  pages={1058--1095},
  year={2022},
  publisher={The Institute of Mathematical Statistics and the Bernoulli Society}
}

@article{ghoshdastidar2020two,
  title={Two-sample hypothesis testing for inhomogeneous random graphs},
  author={Ghoshdastidar, Debarghya and Gutzeit, Maurilio and Carpentier, Alexandra and Von Luxburg, Ulrike},
  journal={The Annals of Statistics},
  volume={48},
  number={4},
  pages={2208--2229},
  year={2020},
  publisher={JSTOR}
}

@article{alyakin2024correcting,
  title={Correcting a nonparametric two-sample graph hypothesis test for graphs with different numbers of vertices with applications to connectomics},
  author={Alyakin, Anton A. and Agterberg, Joshua and Helm, Hayden S. and Priebe, Carey E.},
  journal={Applied Network Science},
  volume={9},
  number={1},
  pages={1},
  year={2024},
  publisher={Springer}
}

@article{du2023hypothesis,
  title={Hypothesis testing for equality of latent positions in random graphs},
  author={Du, Xinjie and Tang, Minh},
  journal={Bernoulli},
  volume={29},
  number={4},
  pages={3221--3254},
  year={2023},
  publisher={Bernoulli Society for Mathematical Statistics and Probability}
}

@article{fan2022simple,
  title={SIMPLE: Statistical inference on membership profiles in large networks},
  author={Fan, Jianqing and Fan, Yingying and Han, Xiao and Lv, Jinchi},
  journal={Journal of the Royal Statistical Society Series B: Statistical Methodology},
  volume={84},
  number={2},
  pages={630--653},
  year={2022},
  publisher={Oxford University Press}
}

@article{levin2025bootstrapping,
  title={Bootstrapping networks with latent space structure},
  author={Levin, Keith and Levina, Elizaveta},
  journal={Electronic Journal of Statistics},
  volume={19},
  number={1},
  pages={745--791},
  year={2025},
  publisher={The Institute of Mathematical Statistics and the Bernoulli Society}
}

@inproceedings{
dilworth2024valid,
title={Valid Bootstraps for Network Embeddings with Applications to Network Visualisation},
author={Emerald Dilworth and Ed Davis and Daniel John Lawson},
booktitle={The 41st Conference on Uncertainty in Artificial Intelligence},
year={2025}
}

@article{holland1983stochastic,
  title={Stochastic blockmodels: First steps},
  author={Holland, Paul W. and Laskey, Kathryn Blackmond and Leinhardt, Samuel},
  journal={Social networks},
  volume={5},
  number={2},
  pages={109--137},
  year={1983},
  publisher={Elsevier}
}

@article{winding2023connectome,
  title={The connectome of an insect brain},
  author={Winding, Michael and Pedigo, Benjamin D. and Barnes, Christopher L. and Patsolic, Heather G. and Park, Youngser and Kazimiers, Tom and Fushiki, Akira and Andrade, Ingrid V. and Khandelwal, Avinash and Valdes-Aleman, Javier and others},
  journal={Science},
  volume={379},
  number={6636},
  pages={eadd9330},
  year={2023},
  publisher={American Association for the Advancement of Science}
}

@article{eschbach2020recurrent,
  title={Recurrent architecture for adaptive regulation of learning in the insect brain},
  author={Eschbach, Claire and Fushiki, Akira and Winding, Michael and Schneider-Mizell, Casey M and Shao, Mei and Arruda, Rebecca and Eichler, Katharina and Valdes-Aleman, Javier and Ohyama, Tomoko and Thum, Andreas S and others},
  journal={Nature Neuroscience},
  volume={23},
  number={4},
  pages={544--555},
  year={2020},
  publisher={Nature Publishing Group US New York}
}

@article{jiang2021models,
  title={Models of heterogeneous dopamine signaling in an insect learning and memory center},
  author={Jiang, Linnie and Litwin-Kumar, Ashok},
  journal={PloS Computational Biology},
  volume={17},
  number={8},
  pages={e1009205},
  year={2021},
  publisher={Public Library of Science San Francisco, CA USA}
}

@article{hoff2002latent,
  title={Latent space approaches to social network analysis},
  author={Hoff, Peter D. and Raftery, Adrian E. and Handcock, Mark S.},
  journal={Journal of the American Statistical Association},
  volume={97},
  number={460},
  pages={1090--1098},
  year={2002},
  publisher={Taylor \& Francis}
}

@article{Zhu06,
	author = {Mu Zhu and Ali Ghodsi},
	journal = {Computational Statistics \& Data Analysis},
	number = {2},
	pages = {918-930},
	title = {Automatic dimensionality selection from the scree plot via the use of profile likelihood},
	volume = {51},
	year = {2006}}

@article{cape19two,
author = {Joshua Cape and Minh Tang and Carey E. Priebe},
title = {{The two-to-infinity norm and singular subspace geometry with applications to high-dimensional statistics}},
volume = {47},
journal = {The Annals of Statistics},
number = {5},
publisher = {Institute of Mathematical Statistics},
pages = {2405 -- 2439},
year = {2019},
doi = {10.1214/18-AOS1752}
}

@article{chen2024spectral,
  title={A spectral-based framework for hypothesis testing in populations of networks},
  author={Chen, Li and Josephs, Nathaniel and Lin, Lizhen and Zhou, Jie and Kolaczyk, Eric D},
  journal={Statistica Sinica},
  volume={34},
  number={1},
  pages={87--110},
  year={2024},
  publisher={JSTOR}
}

@article{liu2020cauchy,
  title={Cauchy combination test: a powerful test with analytic p-value calculation under arbitrary dependency structures},
  author={Liu, Yaowu and Xie, Jun},
  journal={Journal of the American Statistical Association},
  volume={115},
  number={529},
  pages={393--402},
  year={2020},
  publisher={Taylor \& Francis}
}

@article{chung2022valid,
  title={Valid two-sample graph testing via optimal transport procrustes and multiscale graph correlation with applications in connectomics},
  author={Chung, Jaewon and Varjavand, Bijan and Arroyo-Reli{\'o}n, Jes{\'u}s and Alyakin, Anton and Agterberg, Joshua and Tang, Minh and Priebe, Carey E and Vogelstein, Joshua T},
  journal={Stat},
  volume={11},
  number={1},
  pages={e429},
  year={2022},
  publisher={Wiley Online Library}
}

@article{lei2023bias,
  title={Bias-adjusted spectral clustering in multi-layer stochastic block models},
  author={Lei, Jing and Lin, Kevin Z},
  journal={Journal of the American Statistical Association},
  volume={118},
  number={544},
  pages={2433--2445},
  year={2023},
  publisher={Taylor \& Francis}
}

@article{macdonald2022latent,
  title={Latent space models for multiplex networks with shared structure},
  author={MacDonald, Peter W and Levina, Elizaveta and Zhu, Ji},
  journal={Biometrika},
  volume={109},
  number={3},
  pages={683--706},
  year={2022},
  publisher={Oxford University Press}
}

@article{jing2021community,
  title={Community detection on mixture multilayer networks via regularized tensor decomposition},
  author={Jing, Bing-Yi and Li, Ting and Lyu, Zhongyuan and Xia, Dong},
  journal={The Annals of Statistics},
  volume={49},
  number={6},
  pages={3181--3205},
  year={2021},
  publisher={JSTOR}
}

@article{gretton2012kernel,
  title={A kernel two-sample test},
  author={Gretton, Arthur and Borgwardt, Karsten M and Rasch, Malte J and Sch{\"o}lkopf, Bernhard and Smola, Alexander},
  journal={The journal of machine learning research},
  volume={13},
  number={1},
  pages={723--773},
  year={2012},
  publisher={JMLR. org}
}

@ARTICLE{Oselio14,
  author={Oselio, Brandon and Kulesza, Alex and Hero, Alfred O.},
  journal={IEEE Journal of Selected Topics in Signal Processing}, 
  title={Multi-Layer Graph Analysis for Dynamic Social Networks}, 
  year={2014},
  volume={8},
  number={4},
  pages={514-523}}

@article{Durante17,
  author  = {Daniele Durante and Nabanita Mukherjee and Rebecca C. Steorts},
  title   = {Bayesian Learning of Dynamic Multilayer Networks},
  journal = {Journal of Machine Learning Research},
  year    = {2017},
  volume  = {18},
  number  = {43},
  pages   = {1--29}
}

@article{Loyal23,
  author  = {Joshua Daniel Loyal and Yuguo Chen},
  title   = {An Eigenmodel for Dynamic Multilayer Networks},
  journal = {Journal of Machine Learning Research},
  year    = {2023},
  volume  = {24},
  number  = {128},
  pages   = {1--69}
}

@ARTICLE{Wang26,
       title={Multilayer random dot product graphs: Estimation and online change point detection},
  author={Wang, Fan and Li, Wanshan and Madrid Padilla, Oscar Hernan and Yu, Yi and Rinaldo, Alessandro},
  journal={Journal of the Royal Statistical Society Series B: Statistical Methodology},
  volume={88},
  number={1},
  pages={282--312},
  year={2026},
  publisher={Oxford University Press UK}
}





\end{document}